\documentclass[12pt]{article}
\usepackage[usenames,dvipsnames]{color}

\usepackage{graphicx}
\usepackage{amsthm}
\usepackage{graphics}
\usepackage{amsmath,,amssymb}
\usepackage{algorithmic}
\usepackage{algorithm}
\usepackage{caption}
\usepackage{subcaption}
\usepackage{pstricks, pst-node}
\usepackage{tikz}
\usepackage[linewidth=1pt]{mdframed}
\usepackage{multirow}
\usepackage[round,colon,authoryear]{natbib}
\usepackage{pgfplots}
\usepackage{hyperref}
\usepackage{bm}
\usepackage{enumitem}

\pgfplotsset{compat=1.10}
\usepgfplotslibrary{fillbetween}
\usetikzlibrary{patterns}
\usetikzlibrary{shapes,arrows}
\usetikzlibrary{arrows.meta}

\newtheorem{example}{Example}
\newtheorem{theorem}{Theorem}
\newtheorem{proposition}{Proposition}
\newtheorem{lemma}{Lemma}

\newtheorem{assumption}{Assumption}

\def\EE{\mathbb{E}}
\def\PP{\mathbb{P}}
\def\RR{\mathbb{R}}

\def\Bcal{\mathcal B}
\def\Acal{\mathcal A}

\def\bb{\bm b}

\tikzstyle{block} = [rectangle, draw, fill=white!80!black, line width=2pt,
    text width=15em, text centered, rounded corners, minimum height=3em]
\tikzstyle{line} = [draw, -latex',line width=2pt]

\begin{document}
\title{Robust Differential Abundance Test \\in Compositional Data}
\author{Shulei Wang\\ University of Illinois at Urbana-Champaign}
\date{(\today)}

\maketitle

\footnotetext[1]{
	Address for Correspondence: Department of Statistics, University of Illinois at Urbana-Champaign, 725 South Wright Street, 
	Champaign, IL 61820 (Email: shuleiw@illinois.edu).}

\begin{abstract}
Differential abundance tests in compositional data are essential and fundamental tasks in various biomedical applications, such as single-cell, bulk RNA-seq, and microbiome data analysis. However, because of the compositional constraint and the prevalence of zero counts in the data, differential abundance analysis in compositional data remains a complicated and unsolved statistical problem. This study introduces a new differential abundance test, the robust differential abundance (RDB) test, to address these challenges. Compared with existing methods, the RDB test is simple and computationally efficient, is robust to prevalent zero counts in compositional datasets, can take the data's compositional nature into account, and has a theoretical guarantee of controlling false discoveries in a general setting. Furthermore,  in the presence of observed covariates, the RDB test can work with the covariate balancing techniques to remove the potential confounding effects and draw reliable conclusions. Finally, we apply the new test to several numerical examples using simulated and real datasets to demonstrate its practical merits.
\end{abstract}

%\noindent{\bf Keywords:} 

%\noindent{\bf AMS 2000 Subject Classification:} 

\newpage
%%%%%%%%%%%%%%%%%%%%%%%%%%%%%%%%%%%%%%%%%%%%%%%
\section{Introduction}
\label{sc:intro}
%%%%%%%%%%%%%%%%%%%%%%%%%%%%%%%%%%%%%%%%%%%%%%%

Compositional data, where all components are non-negative, and their sum is one, naturally arises in a wide range of modern scientific applications, including human microbiome studies, nutritional science, genomics studies, and geochemistry. An essential and fundamental task in these scientific applications is differential abundance testing, aiming to identify a set of differential components across experimental conditions  \citep{robinson2010edger,law2014voom,kharchenko2014bayesian,fernandes2014unifying,love2014moderated,mandal2015analysis,risso2018general,morton2019establishing}. However, differential abundance analysis in compositional data is a complicated and challenging statistical problem because of the constant-sum constraint. Applying standard statistical analysis directly to compositional data ignores data's compositional nature, and hence can result in an ill-defined statistical hypothesis and false-positive scientific discovery \citep{vandeputte2017quantitative,weiss2017normalization,morton2019establishing,brill2019testing,lin2020analysis}.

To account for the constant-sum constraint and gain reliable scientific insights from compositional data, studies propose many differential abundance testing methods for different types of biomedical compositional datasets, including single-cell, bulk RNA-seq, and microbiome data \citep{robinson2010edger,paulson2013differential,law2014voom,kharchenko2014bayesian,fernandes2014unifying,love2014moderated,mandal2015analysis,le2016mixmc,butler2018integrating,risso2018general,morton2019establishing,lin2020analysis,martin2020modeling}. These methods have been very successful in many applications and helped make significant progress in various scientific fields. However, most of these existing methods require taking a ratio (or log-ratio) of two proportions, and thus, implicitly assume that each component's proportions are strictly larger than zero. Unfortunately, zero counts are prevalent in some biomedical compositional data sets, e.g., microbiome data \citep{cao2020multisample}, because of insufficient sequence depth (also called technical zero) or biological variation (also called structural zero). Using pseudo-counts, that is replacing zero counts with a small positive number, could alleviate this zero counts problem. However, one usually does not know how to choose a pseudo-count for a given data set, and this strategy could potentially lead to inflated false-positive discoveries, as noted by \cite{brill2019testing}. 

These challenges makes one wonder if a differential abundance test exists for compositional data with the following properties: (i) it is simple, easy to implement, and computationally efficient, (ii) is robust to prevalent zero counts in the compositional data set eliminating the need of an extra step to handle the zero counts problem, (iii) can take the compositional nature of data into account, and (iv) has a theoretical guarantee of controlling false discoveries in a general setting. This study shows this is feasible by developing a new robust differential abundance test: the RDB test.

Specifically, motivated by recent studies of \cite{morton2019establishing} and \cite{brill2019testing}, the differential abundance analysis is formulated as a reference-based hypothesis testing problem. That is, the hypothesis is defined with respect to a set of unknown reference components, and its identification conditions are studied. Our investigation shows that under such a reference-based hypothesis, the differential components can be recovered completely by a series of directional comparisons on renormalized proportions in the noiseless case. Inspired by this observation, we introduce a new iterative method to identify the differential components in compositional data. Unlike existing methods, the new method only needs to evaluate standard two-sample $t$ test statistics (also known as Welch's $t$ test) on renormalized proportions in each iteration, so that we do not need to worry about the zero counts problem anymore. Owing to the reference-based hypothesis, the new method follows a two-stage procedure to identify the differential components in each iteration: (i) integrate all components' $t$ test statistics to determine the testing direction; (ii) identify the differential components by a component-wise directional two-sample test. The idea of such a two-stage strategy is inspired by the perspective of empirical Bayesian analysis \citep{robbins1951asymptotically,efron2012large} and has been used in other large-scale simultaneous hypothesis testing problems. For example, \cite{efron2004large} proposes empirically estimating the mean and variance of null distribution before conducting large-scale hypothesis testing. Unlike these existing methods, we iteratively apply such a two-stage strategy. To demonstrate the merit of the newly proposed method, we study the new RDB test's theoretical properties. More concretely, we show that the family-wise error rate (FWER) can be controlled at the $\alpha$ level asymptotically in a general setting, and the differential components can be identified completely in a large probability when the signal-to-noise ratio is large enough. Note that the theoretical analysis itself might be of interest, because of the challenges in handling the dependency among test statistics caused by renormalization and iterative procedure.

As more observational studies collect compositional datasets, another difficulty in differential abundance analysis is reducing the bias introduced by the potential confounding covariates. As the RDB test comprises a series of standard two-sample $t$ tests, it can work with most covariate balancing techniques designed for the potential outcome framework \citep{imbens2015causal,rosenbaum1983central}. In particular, we combine the newly proposed RDB test with weighting methods \citep{rosenbaum1987model,robins2000marginal,imai2014covariate,chan2016globally} to remove the potential confounding effect of observed covariates in this study. Furthermore, the idea of this iterative method in the RDB test can also be extended to control the false discovery rate and test the linear association with a continuous outcome. An R package of RDB is available at https://github.com/lakerwsl/RDB.

%Compared with most existing differential abundance analysis methods, the new strategy does not need to make a model assumption for observed covariates so that it is robust to possible model misspecification. 

%The rest of the paper is organized as follows. We first introduce the model and reference-based hypothesis and give a brief review of existing methods in Section~\ref{sc:model}. In Section~\ref{sc:method}, we present our newly proposed robust differential abundance test and investigate its theoretical properties under a general setting. Section~\ref{sc:cb} discusses the incorporation of the covariate balancing techniques into our newly proposed method. In Section~\ref{sc:nb}, we compare our new method with other common-used methods on both the simulated and real data sets. We finish with concluding remarks in Section~\ref{sc:conclude}. All proofs and auxiliary results are relegated to the Supplementary Material.

%%%%%%%%%%%%%%%%%%%%%%%%%%%%%%%%%%%%%%%%%%%%%%%
\section{Model and Reference-based Hypothesis}
\label{sc:model}
%%%%%%%%%%%%%%%%%%%%%%%%%%%%%%%%%%%%%%%%%%%%%%%

%%%%%%%%%%%%%%%%%%%%%%%%%%
\subsection{Model and Notation}
%%%%%%%%%%%%%%%%%%%%%%%%%%

Suppose we wish to compare the abundance of $d$ different components, such as $d$ different taxa. The absolute abundance of these $d$ different components in a sample can be represented by a non-negative vector $A=(A_1,\ldots,A_d)$, where $A_i\ge 0$ is the absolute abundance of the $i$th component. Instead of directly observing the absolute abundance $A$, we only collect count data $N=(N_1,\ldots, N_d)$ in many real applications,  where $N_i$ is the number of the $i$th component we observe. In this study, given the absolute abundance $A$, we assume the count data are drawn from a multinomial model
\begin{equation}
	\label{eq:model}
	N|P, N^\ast\sim {\rm multinomial}(N^\ast, P),
\end{equation}
where $P=(P_1,\ldots, P_d)$ is the relative abundance of these $d$ components and $N^\ast=\sum_{i} N_i$ is the total counts we observe in a sample. The relative abundance $P$ is defined as $P_i={A_{i} /A^\ast}$, where $A^\ast=\sum_{i=1}^d A_{i}$. As the total counts $N^\ast$ is usually not proportional to absolute abundance, we normalize the count data as a compositional vector, that is, $\hat{P}=(\hat{P}_1,\ldots, \hat{P}_d)$, where $\hat{P}_i=N_i/N^\ast$. Clearly, the compositional data $\hat{P}$ we observe is an empirical version of relative abundance $P$, and thus, only reflects the information on relative abundance rather than absolute abundance. 

In the differential abundance test, we assume two populations of interest, for example, the treated and control group, which can be written as $\pi_1(A)$ and $\pi_2(A)$, respectively. As we want to compare the abundance of the two populations, we draw $m_1$ and $m_2$ samples from each population
$$
A_{k,1},A_{k,2},\ldots, A_{k,m_k}\sim \pi_k(A),\qquad k=1,2.
$$
For the $j$th sample in the $k$th population, we observe count data $N_{k,j}=(N_{k,j,1},\ldots,N_{k,j,d})$,  drawn from the model in \eqref{eq:model} given the absolute abundance $A_{k,j}$. For a fair comparison, these count data can be normalized as empirical relative abundance $\hat{P}_{k,j}=(\hat{P}_{k,j,1},\ldots,\hat{P}_{k,j,d})$. Therefore, the goal of the differential abundance test is to compare the two populations $\pi_1$ and $\pi_2$ based on the observed compositional data $\hat{P}_{k,1},\hat{P}_{k,2},\ldots, \hat{P}_{k,m_k}$, $k=1,2$. Hereafter, we always use $i$ as the index of the component, $j$ as the index of the sample and $k$ as the index of the population. After introducing the data generating model, we also need to define a hypothesis for the differential abundance test. It seems straightforward to define a rigorous hypothesis in such a two-sample case. However, we will see that the data's compositional nature makes it difficult to define a statistically identifiable and scientific meaningful hypothesis. 

The general goal of the differential abundance test is to identify a set of components with different abundances across two populations.  This task can be naturally formulated as a hypothesis testing problem on the absolute abundance
\begin{equation}
	\label{eq:abslt}
	H_{i,0}: \EE_{\pi_1}(A_{i})=\EE_{\pi_2}(A_{i})\qquad {\rm vs.}\qquad H_{i,1}: \EE_{\pi_1}(A_{i})\ne \EE_{\pi_2}(A_{i}),
\end{equation}
where $\EE_{\pi_k}(A_{i})$ represents the expected absolute abundance of the $i$th component in the $k$th population. However, the following example suggests that no methods can consistently test the hypothesis in \eqref{eq:abslt} based on the compositional data, if we do not make further assumptions.
\begin{example}
	\label{ex:abundance}
	Consider the distributions $\pi_1$ and $\pi_2$ in Figure~\ref{fg:toy}. If we examine the absolute abundance of $\pi_1$ and $\pi_2$, we know that the absolute abundances of all components are doubled in $\pi_2$, compared with $\pi_1$. However, only observing their relative abundance, we could conclude that there is no difference between $\pi_1$ and $\pi_2$. In other words, we cannot distinguish the two populations by only observing the relative abundance. 
	%	Suppose $\bA|\pi_1$ has the same distribution with $b\bA|\pi_2$ for some positive number $b\ne 1$, where $\bA|\pi_k$ is a random vector of absolute abundance  drawn from distribution $\pi_k$. Then, the distributions of the relative abundance $P$ are the exactly same under $\pi_1$ and $\pi_2$ although the absolute abundance are different for all components, i.e., $\EE_{\pi_1}(A_{i})\ne \EE_{\pi_2}(A_{i})$ for all $i$. In other words, we cannot distinguish the two populations by only observing compositional data. The distributions $\pi_1$ and $\pi_2$ in Figure~\ref{fg:toy} is an toy example to illustrate this point. 
\end{example}

Example~\ref{ex:abundance} suggests that the information provided by the compositional data is insufficient to answer the question such as \eqref{eq:abslt}. To make the hypothesis identifiable, one possible solution is to ignore the compositional nature of the data and directly test the relative abundance 
\begin{equation}
	\label{eq:relative}
	H_{i,0}: Q_{1,i}=Q_{2,i}\qquad {\rm vs.}\qquad H_{i,1}: Q_{1,i}\ne Q_{2,i},
\end{equation}
where $Q_{1,i}=\EE_{\pi_1}(P_{i})$ and $Q_{2,i}=\EE_{\pi_2}(P_{i})$ are the expected relative abundance of the $i$th component in the two populations. Although such a hypothesis in \eqref{eq:relative} is testable based on the compositional data, it might lead to a false discovery, as it is difficult to find a meaningful connection to absolute abundance. To illustrate this, consider the following example:
\begin{example}
	\label{ex:relative}
	Consider the distributions $\pi_1$ and $\pi_3$ in Figure~\ref{fg:toy}. In absolute abundance, only the red component's absolute abundance is larger in $\pi_3$ than $\pi_1$, and the absolute abundances of the rest components stay the same across the populations. However, examining their relative abundance, we observe that the red component's relative abundance increases and the relative abundances of the rest components decrease due to the compositional constraint. Simply, the changes in relative abundance are not equivalent to the changes in absolute abundance.
	%Suppose $\bA|\pi_1$ has the same distribution with $\bA|\pi_2+t\be_1$, where $t$ is some positive number and $\be_1$ is a vector with only the first component being 1 and all the other components being 0. In other words, only the first component's absolute abundance is larger in the first population than the second population, i.e., $\EE_{\pi_1}(A_{1})>\EE_{\pi_2}(A_{1})$  and $\EE_{\pi_1}(A_{i})=\EE_{\pi_2}(A_{i})$ for all $i=2,\ldots, d$. However, due to the compositional constriant, the relative abundance are different for all components and satisfies $Q_{1,1}> Q_{2,1}$ and $Q_{1,i}< Q_{2,i}$ for all $i=2,\ldots,d$. The distributions $\pi_1$ and $\pi'_2$ in Figure~\ref{fg:toy} is an toy example to illustrate this point. 
\end{example}

In the above example, only the red component is the driver component; however, all the components belong to alternative hypotheses in \eqref{eq:relative}. These two examples indicate that it is important to define an identifiable and interpretable statistical hypothesis carefully for the compositional data. 

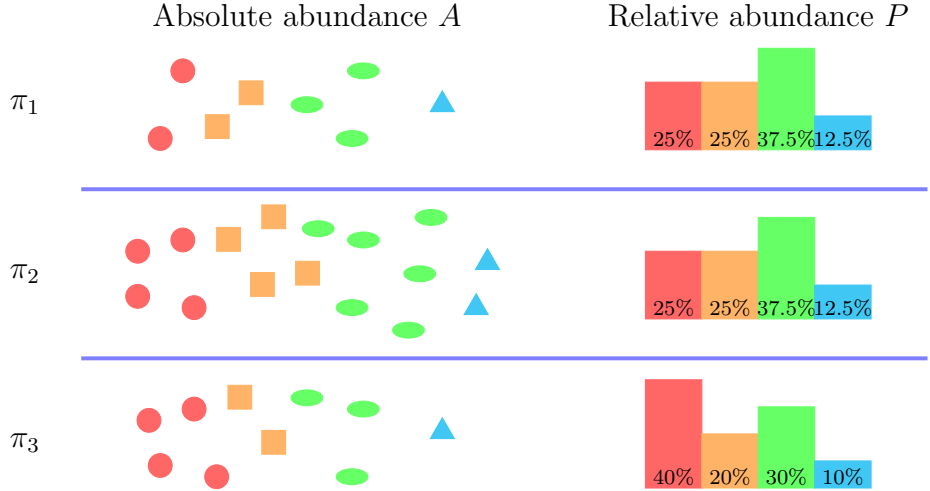
\begin{figure}[h!]
	\centering
	\begin{tikzpicture}[scale=1.5]
		\node[inner sep=0] at (0,1.8) {Absolute abundance $A$};
		\node[inner sep=0] at (4,1.8) {Relative abundance $P$};
		
		\node[inner sep=0] at (-2.5,1) {$\pi_1$};
		
		\filldraw[red!60!white] (-1.1,1.3) circle (3pt);
		\filldraw[red!60!white] (-1.3,0.7) circle (3pt);
		\filldraw[orange!60!white] (-0.6,1) rectangle ++(6pt,6pt);
		\filldraw[orange!60!white] (-0.9,0.7) rectangle ++(6pt,6pt);
		\filldraw[green!60!white] (0,1) ellipse (4pt and 2pt);
		\filldraw[green!60!white] (0.5,1.3) ellipse (4pt and 2pt);
		\filldraw[green!60!white] (0.4,0.7) ellipse (4pt and 2pt);
		\node[fill=cyan!60!white,regular polygon, regular polygon sides=3,inner sep=2pt] at (1.2,1) {};
		
		\filldraw[red!60!white] (3,0.6) rectangle (3.5,1.2);
		\filldraw[orange!60!white] (3.5,0.6) rectangle (4,1.2);
		\filldraw[green!60!white] (4,0.6) rectangle (4.5,1.5);
		\filldraw[cyan!60!white] (4.5,0.6) rectangle (5,0.9);
		
		\node[inner sep=0] at (3.25,0.7) {\scriptsize $25\%$};
		\node[inner sep=0] at (3.75,0.7) {\scriptsize $25\%$};
		\node[inner sep=0] at (4.25,0.7) {\scriptsize $37.5\%$};
		\node[inner sep=0] at (4.75,0.7) {\scriptsize $12.5\%$};
		
		\path[draw,line width=1.5pt,blue!50!white] (-2,0.25) -- (5.5,0.25);
		%\path[draw,line width=1.5pt,blue!50!white] (2.5,1.5) -- (2.5,-1);
		
		\node[inner sep=0] at (-2.5,-0.5) {$\pi_2$};
		
		\filldraw[red!60!white] (-1.1,-0.2) circle (3pt);
		\filldraw[red!60!white] (-1,-0.8) circle (3pt);
		\filldraw[red!60!white] (-1.5,-0.3) circle (3pt);
		\filldraw[red!60!white] (-1.5,-0.7) circle (3pt);
		\filldraw[orange!60!white] (-0.8,-0.3) rectangle ++(6pt,6pt);
		\filldraw[orange!60!white] (-0.5,-0.7) rectangle ++(6pt,6pt);
		\filldraw[orange!60!white] (-0.4,-0.1) rectangle ++(6pt,6pt);
		\filldraw[orange!60!white] (-0.1,-0.6) rectangle ++(6pt,6pt);
		\filldraw[green!60!white] (0.1,-0.1) ellipse (4pt and 2pt);
		\filldraw[green!60!white] (0.5,-0.2) ellipse (4pt and 2pt);
		\filldraw[green!60!white] (0.4,-0.8) ellipse (4pt and 2pt);
		\filldraw[green!60!white] (0.9,-1) ellipse (4pt and 2pt);
		\filldraw[green!60!white] (1,-0.5) ellipse (4pt and 2pt);
		\filldraw[green!60!white] (1.1,0) ellipse (4pt and 2pt);
		\node[fill=cyan!60!white,regular polygon, regular polygon sides=3,inner sep=2pt] at (1.6,-0.4) {};
		\node[fill=cyan!60!white,regular polygon, regular polygon sides=3,inner sep=2pt] at (1.5,-0.8) {};
		
		\filldraw[red!60!white] (3,-0.9) rectangle (3.5,-0.3);
		\filldraw[orange!60!white] (3.5,-0.9) rectangle (4,-0.3);
		\filldraw[green!60!white] (4,-0.9) rectangle (4.5,0);
		\filldraw[cyan!60!white] (4.5,-0.9) rectangle (5,-0.6);
		
		\node[inner sep=0] at (3.25,0.7-1.5) {\scriptsize $25\%$};
		\node[inner sep=0] at (3.75,0.7-1.5) {\scriptsize $25\%$};
		\node[inner sep=0] at (4.25,0.7-1.5) {\scriptsize $37.5\%$};
		\node[inner sep=0] at (4.75,0.7-1.5) {\scriptsize $12.5\%$};
		
		\path[draw,line width=1.5pt,blue!50!white] (-2,0.25-1.5) -- (5.5,0.25-1.5);
		
		\node[inner sep=0] at (-2.5,-2) {$\pi_3$};
		
		\filldraw[red!60!white] (-1,-1.7) circle (3pt);
		\filldraw[red!60!white] (-0.8,-2.3) circle (3pt);
		\filldraw[red!60!white] (-1.4,-1.8) circle (3pt);
		\filldraw[red!60!white] (-1.3,-2.2) circle (3pt);
		\filldraw[orange!60!white] (-0.7,-0.2-1.5) rectangle ++(6pt,6pt);
		\filldraw[orange!60!white] (-0.4,-0.6-1.5) rectangle ++(6pt,6pt);
		\filldraw[green!60!white] (0,-0.1-1.5) ellipse (4pt and 2pt);
		\filldraw[green!60!white] (0.5,-0.2-1.5) ellipse (4pt and 2pt);
		\filldraw[green!60!white] (0.4,-0.8-1.5) ellipse (4pt and 2pt);
		\node[fill=cyan!60!white,regular polygon, regular polygon sides=3,inner sep=2pt] at (1.2,-0.4-1.5) {};
		
		\filldraw[red!60!white] (3,-0.9-1.5) rectangle (3.5,0.06-1.5);
		\filldraw[orange!60!white] (3.5,-0.9-1.5) rectangle (4,-0.42-1.5);
		\filldraw[green!60!white] (4,-0.9-1.5) rectangle (4.5,-0.18-1.5);
		\filldraw[cyan!60!white] (4.5,-0.9-1.5) rectangle (5,-0.66-1.5);	
		
		\node[inner sep=0] at (3.25,0.7-3) {\scriptsize $40\%$};
		\node[inner sep=0] at (3.75,0.7-3) {\scriptsize $20\%$};
		\node[inner sep=0] at (4.25,0.7-3) {\scriptsize $30\%$};
		\node[inner sep=0] at (4.75,0.7-3) {\scriptsize $10\%$};			
		
	\end{tikzpicture}
	\caption{Comparison between absolute abundances and relative abundances.}\label{fg:toy}
\end{figure}

%%%%%%%%%%%%%%%%%%%%%%%%%%
\subsection{Reference-based Hypothesis}
%%%%%%%%%%%%%%%%%%%%%%%%%%

As every piece of information introduced by compositional data is relative, reference frames are necessary to analyze the compositional data, as noted in \cite{pawlowsky2011compositional} and \cite{morton2019establishing}. The idea of reference frames has also been applied in standard compositional data analysis. For example, the geometric mean of all components is seen as a reference component in the center log-ratio transformation. A change in compositional data can be interpreted as a change with respect to the reference components through the reference frames.

Specifically, motivated by recent work \citep{brill2019testing}, a subset of components $I_0$ is called a reference set if there exists a positive constant $b>0$ such that the relative abundance in $I_0$ changes in the same amplitude  
\begin{equation}
	\label{eq:reference}
	Q_{1,i}=bQ_{2,i},\qquad i\in I_0.
\end{equation}
The reference set's definition suggests that the relative relationship within the reference set $I_0$ does not change across the two populations, that is, $Q_{1,i}/Q_{1,i'}=Q_{2,i}/Q_{2,i'}$ for any $i,i'\in I_0$ so that it can be seen as a benchmark. For example, the orange, green, and blue components in Example~\ref{ex:relative} can be seen as a reference set, although their relative abundances differ between $\pi_1$ and $\pi_3$. Based on the reference set $I_0$, we compare everything with the components in the reference set, and thus, the problem of testing for differential components can be cast as the following reference-based hypothesis testing problem:
\begin{equation}
	\label{eq:reftest}
	H_{i,0}: {Q_{1,i}\over \sum_{i\in I_0}Q_{1,i}}={Q_{2,i} \over \sum_{i\in I_0}Q_{2,i}}\qquad {\rm vs.}\qquad H_{i,1}: {Q_{1,i}\over \sum_{i\in I_0}Q_{1,i}}\ne {Q_{2,i} \over \sum_{i\in I_0}Q_{2,i}}.
\end{equation}
Under the above hypothesis, a change in some components is interpreted as a change with respect to the reference set $I_0$. Unlike the conventional statistical hypothesis testing problem, the reference-based hypothesis at each component is defined by all components. If a component belongs to the reference set, then the null hypothesis is naturally true. The definition in \eqref{eq:reference} also defines a set of non-differential components. If we further assume that the absolute abundance within the reference set is unchanged across two populations, that is, 
\begin{equation}
	\label{eq:absequal}
	\EE_{\pi_1}(A_{i})=\EE_{\pi_2}(A_{i}),\qquad i\in I_0,
\end{equation}
then hypothesis \eqref{eq:abslt} is roughly equivalent to hypothesis \eqref{eq:reftest}. In other words, under the assumption \eqref{eq:absequal}, the compositional data can be used to infer the differential components defined based on absolute abundance. Note also that assumption \eqref{eq:absequal} cannot be diagnosed and verified based on the compositional data alone.

The definition of hypothesis \eqref{eq:reftest} depends on the choice of the reference set. If domain knowledge or extra information on the reference set is available in advance, \eqref{eq:reftest} is well-defined based on the known $I_0$ and ready to be tested by the compositional data. However, a more realistic situation in practice is that the reference set is unknown in advance. In this case, different reference set choices can lead to inconsistent conclusions of the null hypothesis, and thus, cause the problem of model identifiability. The following proposition shows that assumptions are required to make the hypothesis testing problem identifiable.
\begin{proposition}
	\label{prop:indetification}
	If we assume the reference set $I_0$ in \eqref{eq:reftest} satisfies $|I_0|>d/2$, then the null and alternative hypotheses in \eqref{eq:reftest} are well defined. In contrast, there exists an instance of two reference sets $I_{0,1}$ and $I_{0,2}$ defined in \eqref{eq:reference} with $|I_{0,1}|,|I_{0,2}|\le d/2$ and the null hypotheses defined by $I_{0,1}$ and $I_{0,2}$ contradict each other.
\end{proposition}

This proposition suggests that when the reference set is unknown, it is necessary to assume the existence of a large enough reference set to ensure that the hypothesis for the differential abundance test in compositional data is well defined. Hereafter, we always assume the reference set $I_0$ in \eqref{eq:reftest} satisfies $|I_0|>d/2$. 

%%%%%%%%%%%%%%%%%%%%%%%%%%
\subsection{Existing Methods for Reference-based Hypothesis}
\label{sc:existing}
%%%%%%%%%%%%%%%%%%%%%%%%%%

%To better demonstrate the idea of these methods, we assume $Q_{k,i}$ is known for $k=1,2$ and $i=1,\ldots, d$ in this section.

To test hypothesis \eqref{eq:reftest}, several different methods have been proposed using standard compositional analysis techniques.  \cite{mandal2015analysis} recently proposed a method called analysis of composition of microbiome (ANCOM). To test if the $i$th component is a differential one, ANCOM compares the $i$th component with all other components by testing
$$
H_{i,i',0}:\EE_{\pi_1}\left\{\log(P_{i}/P_{i'})\right\}=\EE_{\pi_2}\left\{\log(P_{i}/P_{i'})\right\}\quad {\rm vs}\quad H_{i,i',1}:\EE_{\pi_1}\left\{\log(P_{i}/P_{i'})\right\}\ne \EE_{\pi_2}\left\{\log(P_{i}/P_{i'})\right\}
$$
for all $i'\ne i$. After $d(d-1)/2$ hypothesis testing, ANCOM summarizes all the decisions by the number of null hypothesis rejections, that is,
$$
W_i=\sum_{i'\ne i}I[\EE_{\pi_1}\left\{\log(P_{i}/P_{i'})\right\}\ne \EE_{\pi_2}\left\{\log(P_{i}/P_{i'})\right\}],
$$
where $I(\cdot)$ is an indicator function. The $i$th component is a differential one if $W_i>d/2$.  The main disadvantage of this method is that comparing $O(d^2)$ hypotheses is time-consuming in practice. With a similar idea, \cite{brill2019testing} proposes to first identify the reference set by evaluating
$$
S_i={\rm Median}\left[|\EE_{\pi_1}\left\{\log(P_{i}/P_{i'})\right\}-\EE_{\pi_2}\left\{\log(P_{i}/P_{i'})\right\}|:i'\ne i\right].
$$
The $i$th component belongs to the reference set if $S_i$ is small. After the reference set is estimated, the differential components can be identified by comparing each component with the estimated reference set. However, identifying the reference set effectively and the effect of a misspecified reference set on this method remain unknown. Unlike the above two methods, \cite{morton2019establishing} proposes to rank the components by the ratio between two populations. Specifically, the ratio of $i$th components is defined as
$$
T_i=\EE_{\pi_1}\{\log(P_{i})\}-\EE_{\pi_2}\{\log(P_{i})\}.
$$
Then, one can find the mode of $T_i$, which is defined as $T^\ast$ such that $|\{i:T_i=T^\ast\}|$ is the largest.  Then, the differential components are defined as the ones with $T_i\ne T^\ast$. Nevertheless, it is difficult to estimate the mode of $T_i$ when there is noise.

In all the above methods, the important tool to account for the compositional nature is the ratio or log-ratio of two proportions. This idea is also commonly used in the existing literature of compositional data analysis \citep{aitchison1983principal,pawlowsky2011compositional}. Its main advantage is that it can represent the relative information in a clean form and can be easily incorporated into various classical multivariate analyses when the proportions are strictly larger than 0. However, the ratio between the two proportions can be unstable and even ill-defined in the presence of zero counts and measurement error. Unfortunately, zero counts are prevalent in many compositional data sets, such as microbiome data sets. To overcome this problem, a popular practice is to replace zero counts with a small number (also called pseudo-count). However, choosing this small number for a given data set remains a challenge. In addition, as \cite{brill2019testing} show, mishandling zero counts can lead to inflated false discoveries.

\section{Robust Differential Abundance Test}
\label{sc:method}
%%%%%%%%%%%%%%%%%%%%%%%%%%%%%%%%%%%%%%%%%%%%%%%

%%%%%%%%%%%%%%%%%%%%%%%%%%
\subsection{Infinite Sample Size}
%%%%%%%%%%%%%%%%%%%%%%%%%%

The analysis in the previous section might raise questions if there is a method to test the differential components in compositional data, which is robust to prevalent zero counts in compositional data. To answer this question, we introduce a new robust framework for differential abundance test in compositional data, which only examines the absolute difference between two proportions. Without loss of generality, we always write the largest reference set as $I_0$, which means $H_{i,1}$ in \eqref{eq:reftest} (based on $I_0$) is true as long as $i\notin I_0$, and the set of differential components as $I_1=[d]\setminus I_0$, where $[d]=\{1,\ldots, d\}$. Based on this definition, $I_0$ is essentially the set of all non-differential components; therefore we might use these two concepts exchangeably. 

%Instead of evaluating the ratio, we consider a new iterative strategy to identify differential components by looking at the absolute difference between two proportions.

Compared with conventional hypothesis testing, the main difficulty in testing a reference-based hypothesis is that the null hypothesis of the $i$th component cannot be identified only based on data at the $i$th component, as \eqref{eq:reftest} is defined based on the relative relationship with respect to components in $I_0$. This section shows that testing such hypotheses can be achieved by a series of directional comparisons on renormalized proportions. To illustrate the idea, we consider the ideal case that the number of observed samples is infinite, that is, $Q_{1,i}$ and $Q_{2,i}$  are known, and the sum of the components between the two conditions is strictly larger than 0, that is, $Q_{1,i}+Q_{2,i}>0$ for $i\in[d]$ in this section. 

According to \eqref{eq:reference}, a common feature of all non-differential components (in $I_0$) is that $R_i:=Q_{1,i}-Q_{2,i}$ has the same sign for $i\in I_0$. More concretely, for all $i\in I_0$, $R_i>0$ if $b>1$; $R_i<0$ if $b<1$; and $R_i=0$ if $b=1$. This observation motivates us to use all components to infer the sign of $b-1$ and then identify the differential components by comparing the signs of $R_i$ and $b-1$. In particular, the assumption $|I_0|>d/2$ suggests that the sign of all $R_i$'s median, denoted by $M\{R_i:1,\ldots,d\}$, always equals to the sign of $R_i$ for $i\in I_0$. Therefore, $b-1$ has the same sign as $M\{R_i:1,\ldots,d\}$. As all components with a different sign of $b-1$ are differential components, we can identify differential components by comparing the sign of $R_i$ and $M\{R_i:1,\ldots,d\}$. Briefly, all components with ${\rm sign}(R_i)\ne {\rm sign}(M\{R_i:1,\ldots,d\})$ are differential components. The differential components identified by the above strategy are all correct ones, but it might ignore many other differential components. For example, if $b>1$, the above strategy cannot identify components with $Q_{1,i}>bQ_{2,i}$, although they too are also differential components based on the definition in \eqref{eq:reftest}. 

To overcome this problem, we generalize the above idea by repeatedly applying this strategy to renormalized proportions. To be specific, let $U_{(t)}$ be the set of identified differential components candidates and $V_{(t)}$ be the rest of the components at each iteration $t$. We set $U_{(0)}=\emptyset$ and $V_{(0)}=[d]$ to initialize the iterative procedure. Instead of assessing $R_i$ directly, we opt to examine the absolute difference after normalization with respect to a set $I\subset [d]$
$$
R_i(I)={Q_{1,i} \over \sum_{i\in I}Q_{1,i}}-{Q_{2,i} \over \sum_{i\in I}Q_{2,i}}, \qquad i\in I.
$$
If we define $b(I)=b\times (\sum_{i\in I}Q_{2,i}/\sum_{i\in I}Q_{1,i})$, we can know that $b(I)$ is the ratio of $Q_{1,i}/ \sum_{i\in I}Q_{1,i}$ and $Q_{2,i}/ \sum_{i\in I}Q_{2,i}$ for all $i\in I_0$, and $R_i(I)$'s median has the same sign with $b(I)-1$ when $I_0\subset I$. Here, we define the median of $R_i(I)$ as
$$
M(I)={\rm Median}\left\{R_i(I):i\in I\right\}.
$$
After determining $M(I)$'s sign, we can define the set of differential components $W^{+}(I)=\left\{i\in I:R_i(I)> 0\right\}$, $W^{-}(I)=\left\{i\in I: R_i(I)<0\right\}$ and $W^{o}(I)=\left\{i\in I:R_i(I)= 0\right\}$.

Following these notations, we now define the following iterative procedure to identify the differential components. For any $t=0,1,\ldots$, we repeat the following procedure:
\begin{enumerate}[label=(\alph*)]
	\item Find the median $M(V_{(t)})$ based on $R_i(V_{(t)})$, $i\in V_{(t)}$.
	\item Find all components with a different sign from $M(V_{(t)})$
	$$
	W_{(t)}=\begin{cases}
		W^{+}(V_{(t)})\cup W^{-}(V_{(t)}), \quad & {\rm if}\ M(V_{(t)})=0\\
		W^{-}(V_{(t)})\cup W^{o}(V_{(t)}), \quad & {\rm if}\ M(V_{(t)})>0\\
		W^{+}(V_{(t)})\cup W^{o}(V_{(t)}), \quad & {\rm if}\ M(V_{(t)})<0
	\end{cases}.
	$$
	\item Let $U_{(t+1)}=U_{(t)}\cup W_{(t)}$ and $V_{(t+1)}=V_{(t)}\setminus  W_{(t)}$. If $W_{(t)}=\emptyset$, stop the loop.
\end{enumerate}
After the loop stop at $t=T$, we set $\hat{I}_0=V_{(T)}$ and $\hat{I}_1=U_{(T)}$. $\hat{I}_0$ is an estimation of  the largest reference set $I_0$ and $\hat{I}_1$ estimates the set of differential components $I_1$. We now show that $I_0$ and $I_1$ can be recovered by $\hat{I}_0$ and $\hat{I}_1$ completely in at most $|I_1|+1$ iterations. 
\begin{theorem}
	\label{thm:nonoise}
	Consider estimating $I_0$ and $I_1$ by $\hat{I}_0$ and $\hat{I}_1$ defined above. If $|I_0|>d/2$, we can show that
	$$
	T\le |I_1|+1, \qquad \hat{I}_0=I_0\qquad {\rm and }\qquad \hat{I}_1=I_1.
	$$ 
\end{theorem}

This theorem suggests that one can identify the differential components accurately by a series of directional comparisons on renormalized proportion, when the sample size is infinite. Compared with existing methods, this iterative procedure does not evaluate the ratio between two proportions; therefore we need not worry about the prevalent zero counts problem anymore.

%%%%%%%%%%%%%%%%%%%%%%%%%%
\subsection{Finite Sample Size}
\label{sc:finite}
%%%%%%%%%%%%%%%%%%%%%%%%%%

We cannot apply the iterative procedure described in the previous section directly to the compositional data in practice, because we only observe finite samples. Owing to the uncertainty in data, we redefine the testing procedure in the previous section to account for the randomness in data. Given a subset $I\subset [d]$, we consider a standard two-sample $t$ test statistic after normalization with respect to $I$
$$
\hat{R}_i(I)={\bar{P}_{1,i}(I)-\bar{P}_{2,i}(I) \over [\hat{\sigma}^2_{1,i}(I)/m_1+\hat{\sigma}^2_{2,i}(I)/m_2]^{1/2}},
$$
where $\bar{P}_{k,i}(I)$ and $\hat{\sigma}^2_{k,i}(I)$ are defined as $\bar{P}_{k,i}(I)={\bar{P}_{k,i}/ \sum_{i\in I}\bar{P}_{k,i}}$, where $\bar{P}_{k,i}=\sum_{j=1}^{m_k}\hat{P}_{k,j,i}/m_k$,
and 
$$
\hat{\sigma}^2_{k,i}(I)={1\over m_k-1}\sum_{j=1}^{m_k}\left({\hat{P}_{k,j,i}\over \sum_{i\in I}\bar{P}_{k,i}}-\bar{P}_{k,i}(I)\right)^2.
$$
Similar to $M(I)$, we also define $\hat{R}_i(I)$'s median, that is,
$$
\hat{M}(I)={\rm Median}\left\{\hat{R}_i(I):i\in I\right\}.
$$

Suppose $\hat{M}(I)$ is very close to 0 in a sense that $-M<\hat{M}(I)<M$, we can expect that $b(I)$ is close to 1 and all components with $\hat{R}_i(I)$ far away from $0$ are differential components. This task is equivalent to conducting two-sided hypothesis testing, and we define the set $W^{\pm}(I,D)=\{i\in I: |\hat{R}_i(I)|>D \}$, where $D$ is a critical value corrected for the multiple testing. In contrast, if $\hat{M}(I)$ is large than $M$, we can expect $b(I)$ is larger than 1, and the differential components can be identified by one-sided hypothesis testing, that is, $W^{-}(I,D)=\{i\in I: \hat{R}_i(I)<-D \}$. Similarly, when $\hat{M}(I)<-M$, we collect the differential components in the set $W^{+}(I,D)=\{i\in I: \hat{R}_i(I)>D \}$.

Similar to the previous section, we define $U_{(t)}$, as the set of selected differential components and $V_{(t)}$, as the rest of the components at each iteration $t$, which are initially chosen as $U_{(0)}=\emptyset$ and $V_{(0)}=[d]$. Given a threshold $M$ and a series of critical values $(D^{\pm}_{(t)},D^{+}_{(t)},D^{-}_{(t)})$, $t=0,1,\ldots$, we update $U_{(t)}$ and $V_{(t)}$ in the following rule. 
\begin{enumerate}[label=(\alph*)]
	\item Find the median $\hat{M}(V_{(t)})$ based on $\hat{R}_i(V_{(t)})$, $i\in V_{(t)}$.
	\item Set
	$$
	W_{(t)}=\begin{cases}
		W^{\pm}(V_{(t)},D^{\pm}_{(t)}), \quad & {\rm if}\ -M<\hat{M}(V_{(t)})<M\\
		W^{-}(V_{(t)},D^-_{(t)}), & {\rm if}\ \hat{M}(V_{(t)})\ge M\\
		W^{+}(V_{(t)},D^+_{(t)}), & {\rm if}\ \hat{M}(V_{(t)})\le -M
	\end{cases}.
	$$
	\item Let $U_{(t+1)}=U_{(t)}\cup W_{(t)}$ and $V_{(t+1)}=V_{(t)}\setminus  W_{(t)}$. If $W_{(t)}=\emptyset$, stop the loop.
\end{enumerate}
After the loop stop at $t=T$, we set $\hat{I}_0=V_{(T)}$ and $\hat{I}_1=U_{(T)}$. To implement this procedure, we still need to choose a threshold $M$ and a sequence of positive critical values $(D^{\pm}_{(t)},D^{+}_{(t)},D^{-}_{(t)})$ for $t=0,1,\ldots$ to control false discovery.

%%%%%%%%%%%%%%%%%%%%%%%%%%
\subsection{Family-wise Error Rate Control}
\label{sc:fwer}
%%%%%%%%%%%%%%%%%%%%%%%%%%

Unlike the conventional false discovery setting, there are several unique challenges here. First, the test statistics $\hat{R}_i(I)$ are generally dependent for different $i\in I$ because of the compositional constraint after renormalization. In addition, there is sometimes some dependency structure between components in some applications, such as microbiome data \citep{hawinkel2019broken}. Second, the two-step procedure suggests that we need to consider the potential error in the first step when choosing the second step's critical values. Third, as an iterative procedure, the possible dependence after renormalization plays an important role in the false discovery control.

To choose the threshold $M$ for the median, we need to consider the potential dependency structure between $\hat{R}_i(I)$. In particular, let $P_i=\{\hat{P}_{k,j,i}, k=1,2,j=1,\ldots,m_k\}$ and $S_l$ be a subset of $[0,1]^{m_1+m_2}$ for all $1\le l\le |I_0|$. For any subset $S_l$, we assume $P_i, i\in I_0$ can be ordered as $P_{(1)},\ldots, P_{(|I_0|)}$  to satisfy
\begin{equation}
	\label{eq:dependency}
	\sum_{l=1}^{|I_0|}\|\EE((B_l-\mu_l)^2|B_1^{l-1})\|_\infty+2\sum_{l=1}^{|I_0|}\sum_{l'=1}^{l-1}\|\EE(B_l-\mu_l|B_1^{l'})\|_\infty\le C_p\sum_{l=1}^{|I_0|}\mu_l,
\end{equation}
where $C_p>0$ is some constant, $B_l=I(P_{(l)}\in S_l)$, $\mu_l=\EE(B_l)$, and $B_{1}^l=\{B_1,\ldots,B_l\}$. We can set $C_p=1$ when $P_i$ are independent of each other. Given condition \eqref{eq:dependency}, we choose the median's threshold $M$ as $M=\sqrt{{C_p\log d/d}}$. In doing so, we can ensure that $\PP(\hat{M}(I)<-M,b(I)\ge 1)\to 0$ and $\PP(\hat{M}(I)>M,b(I)\le 1)\to 0$.

When choosing the critical values $(D^{\pm}_{(t)},D^{+}_{(t)},D^{-}_{(t)})$, it is necessary to consider renormalization at each iteration. Specifically, let $q_\alpha$ be the upper $\alpha$-quantile of the following random variable $\max_{i\in I_0}(\sup_{r_i}r_iz_{1,i}+\sqrt{1-r_i^2}z_{2,i})=\max_{i\in I_0}\sqrt{z_{1,i}^2+z_{2,i}^2}$,
where $z_{1,i}$ and $z_{2,i}$ are independent standard Gaussian random variables. For simplicity, we can choose an upper bound of $q_\alpha$ in practice, that is, $\tilde{q}_\alpha=\sqrt{2\log d-2\log \alpha}$. In one-sided testing, we choose $D^{+}_{(t)}=D^{-}_{(t)}=q_\alpha$. When choosing a critical value in two-sided testing, it is also necessary to consider the effect of potential error in the first step. In particular, we choose critical value $D^{\pm}_{(t)}$ as
$$
D^{\pm}_{(t)}=q_\alpha+r_QM,
$$
where $r_Q={8C_rU_f/L_{1/4}}$ and $U_f$, $L_{1/4}$ and $C_r$ are defined in the Supplementary Material. When $r_Q\sqrt{C_p}\ll \sqrt{d}$, we have $D^{\pm}_{(t)}=q_\alpha(1+o(1))$. Based on the above thresholds and critical values, we show that the proposed procedure can asymptotically control the family-wise error rate (FWER) at the $\alpha$ level. 

\begin{theorem}
	\label{thm:fwer}
	Consider estimating $I_0$ and $I_1$ by $\hat{I}_0$ and $\hat{I}_1$ defined in Section~\ref{sc:finite}. If \eqref{eq:dependency} is satisfied and Assumptions~\ref{assmp:center}-\ref{assmp:signal} in the Supplementary Material hold, we can show that
	$$
	\PP(\hat{I}_1\cap I_0\ne \emptyset)\le \alpha+o(1).
	$$ 
\end{theorem}

Theorem~\ref{thm:fwer} suggests that FWER can be controlled at the $\alpha$ level asymptotically by our newly proposed method. Section~\ref{sc:fdr} shows that the critical values $(D^{\pm}_{(t)},D^{+}_{(t)},D^{-}_{(t)})$ can also be chosen to control the false discovery rate (FDR). 

In the RDB test, we need to specify three parameters: the threshold for median $M$, the critical value for each component $T$, and the effect size of the first step $r_Q$ if we set $D^{\pm}_{(t)}=T+r_QM$ and $D^{+}_{(t)}=D^{-}_{(t)}=T$. $M$ and $r_Q$ mainly control the decision of testing direction in each iteration and its effect. Although larger $M$ and $r_Q$ increases robustness of the test in testing the wrong direction in the correlated case, they lead to smaller power in the second step. The choices of $M=\sqrt{2\log d/d}$ and $r_Q=0.2$ are supported by our experience and used in all numerical experiments of this paper. The simulation results suggest that the performance of RDB is not very sensitive to the choices of $M$ and $r_Q$, even in the correlated case. However, we need to carefully choose the critical value $T$, as it plays an important role in our test. The simulation experiments suggest that the theoretical choice $T=\sqrt{2\log d-2\log \alpha}$ works well; however, as we illustrate in Section~\ref{sc:fdr}, $T$ can also be chosen automatically, when we aim to control FDR.

%%%%%%%%%%%%%%%%%%%%%%%%%%
\subsection{Power Analysis}
%%%%%%%%%%%%%%%%%%%%%%%%%%

To conduct power analysis, we consider a three-component mixture model. Specifically, we assume the absolute abundances of each component $A_i$ are independent of each other, and $\{1,\ldots,d\}$ is split into three sets, $I_0$, $I_1^+$, and $I_1^-$. If the component is non-differential, that is, $i\in I_0$, we assume $A_i$ has the same distribution under $\pi_1$ and $\pi_2$. In contrast, if $i\in I_1^+$, we assume $A_{i}$ under $\pi_2$ has the same distribution with $A_{i}+\delta_+$ under $\pi_1$. Similarly, $A_{i}$ under $\pi_2$ has the same distribution with $A_{i}-\delta_-$ under $\pi_1$ if $i\in I_1^-$. Here, $\delta_+$ and $\delta_-$ are two constants. Under this model, $I_1=I_1^-\cup I_1^+$ is the set of differential components. We assume the sample size in the two groups is the same $m=m_1=m_2$. We also assume the variance of $\hat{P}_{k,j,i}$ to be the same for $k=1,2$, which is denoted by $\sigma_i^2$. We now characterize the sufficient conditions for reliably identifying differential components. 

\begin{theorem}
	\label{thm:power}
	Suppose $\delta_+$ and $\delta_-$ satisfy
	$$
	\EE_{\pi_1}\left({\delta_+-A_i|I_1^+|\delta_+/\sum_{i\in [d]}A_i \over \sum_{i\in [d]}A_i+\delta^\ast}\right)\ge (2+\epsilon)\sigma_{i} \sqrt{2\log d \over m},\qquad i\in I_1^+
	$$
	and
	$$
	\EE_{\pi_1}\left({A_i|I_1^-|\delta_-/\sum_{i\in [d]}A_i -\delta_-\over \sum_{i\in [d]}A_i+\delta^\ast}\right)\ge (2+\epsilon)\sigma_{i} \sqrt{2\log d \over m},\qquad i\in I_1^-,
	$$
	where $\delta^\ast=|I_1^+|\delta_+-|I_1^-|\delta_-$ and $\epsilon$ is some small constant. If we assume Assumption~\ref{assmp:samplesize}, \ref{assmp:signal} and $r_Q\sqrt{C_p}\ll \sqrt{d}$ hold, we have
	$$
	\PP(I_1\subset \hat{I}_1)\to 1.
	$$
\end{theorem}

Theorem~\ref{thm:power} suggests that if the difference between two groups is sufficiently large (signal-to-noise ratio is at rate $\sqrt{\log d/m}$), the proposed methods can identify differential components consistently.

%%%%%%%%%%%%%%%%%%%%%%%%%%%%%%%%%%%%%%%%%%%%%%%
\section{Extensions of Robust Differential Abundance Test}
\label{sc:cb}
%%%%%%%%%%%%%%%%%%%%%%%%%%%%%%%%%%%%%%%%%%%%%%%

%%%%%%%%%%%%%%%%%%%%%%%%%%%%%%%%%%%%%%%%%%%%%%%
\subsection{Covariate Balancing}
%%%%%%%%%%%%%%%%%%%%%%%%%%%%%%%%%%%%%%%%%%%%%%%

Besides the compositional data from the two populations, we usually observe several covariates for each sample in observational studies. Specifically, our observed data is $(\hat{P}_{k,1},X_{k,1}),\ldots,(\hat{P}_{k,m_k},X_{k,m_k})$ for $k=1,2$, where $X_{k,j}\in \RR^d$ are the observed covariates. The imbalance between distributions of $X_{1,j}$ and $X_{2,j}$ can bring bias when covariates are related to both treatment assignment and outcome of interest \citep{imbens2015causal,rosenbaum1983central}. To remove the potential confounding effect, we adopt covariate balancing techniques. As the proposed robust differential abundance test is composed of a series of two-sample $t$ tests, we adopt one of the most widely used covariate balancing methods, that is, a weighting methods, to reduce the bias introduced by the confounding effects \citep{rosenbaum1987model,robins2000marginal,chan2016globally,yu2020covariate}.  Specifically, we adopt a weighting method to estimate weights $w_{k,j}$ for each sample by the observed covariates $X_{k,j}$ for $k=1,2$ and $j=1,\ldots, m_k$. Here, we assume $\sum_{j=1}^{m_k}w_{k,j}=1$ for $k=1,2$. Instead of a standard two-sample $t$ test, we consider a weighted version of the two-sample $t$ test for a subset $I\subset [d]$
$$
\hat{R}_{i,w}(I)={\bar{P}_{1,i,w}(I)-\bar{P}_{2,i,w}(I) \over  \{\hat{V}_{i,w}(I)\}^{1/2}},
$$
where $\bar{P}_{k,i,w}(I)$ and $\hat{\sigma}^2_{k,i}(I)$ are defined as
$$
\bar{P}_{k,i,w}(I)={\bar{P}_{k,i,w} \over \sum_{i\in I}\bar{P}_{k,i,w}}\qquad {\rm where} \qquad \bar{P}_{k,i,w}=\sum_{j=1}^{m_k}w_{k,j}\hat{P}_{k,j,i},
$$
and $\hat{V}_{i,w}(I)$ is an estimator for the variance of $\bar{P}_{1,i,w}(I)-\bar{P}_{2,i,w}(I) $. The statistic $\hat{R}_{i,w}(I)$ can then replace $\hat{R}_i(I)$ in the proposed method in Section~\ref{sc:finite}. Note that we can use any weighting method as long as we can estimate the variance. In the next section's numerical experiments, we consider the empirical balancing calibration weighting method (CAL) proposed by \cite{chan2016globally} (implemented in R package \texttt{ATE}), as it provides an estimator of the variance.

%\begin{algorithm}[h!]
%	\caption{Robust Differential Abundance (RDB) Test}
%	\label{ag:cbrda}
%	\begin{algorithmic}
	%		\REQUIRE Data $(\bX_{k,j},\hat{P}_{k,j})$ for $j=1,\ldots,m_k$ and $k=1,2$.
	%		\ENSURE A set of identified differential components $\hat{I}_1$.
	%		\STATE Construct the weights $w_{k,j}$ by $\bX_{k,j}$.
	%		\STATE Set $U_{(0)}=\emptyset$ and $V_{(0)}=[d]$.
	%		\FOR {$t= 0, 1,\ldots$}
	%		\STATE Find $\hat{M}(V_{(t)})$ based on $\hat{R}_{i,w}(V_{(t)})$.
	%		\STATE Set
	%		$$
	%		W_{(t)}=\begin{cases}
		%			W^{\pm}(V_{(t)},D^{\pm}_{(t)}), \quad & {\rm if}\ -M<\hat{M}(V_{(t)})<M\\
		%			W^{-}(V_{(t)},D^-_{(t)}), & {\rm if}\ \hat{M}(V_{(t)})\ge M\\
		%			W^{+}(V_{(t)},D^+_{(t)}), & {\rm if}\ \hat{M}(V_{(t)})\le -M
		%		\end{cases}.
	%		$$
	%		\STATE Let $U_{(t+1)}=U_{(t)}\cup W_{(t)}$ and $V_{(t+1)}=V_{(t)}\setminus  W_{(t)}$. If $W_{(t)}=\emptyset$, stop the loop.
	%		\ENDFOR
	%		\STATE After the loop stop at $t=T$, we set $\hat{I}_1=U_{(T)}$. 
	%		\RETURN The set $\hat{I}_1$.
	%	\end{algorithmic}
%\end{algorithm}

%%%%%%%%%%%%%%%%%%%%%%%%%%%%%%%%%%%%%%%%%%%%%%%
\subsection{False Discovery Rate Control}
\label{sc:fdr}
%%%%%%%%%%%%%%%%%%%%%%%%%%%%%%%%%%%%%%%%%%%%%%%

The choices of $D^{+}_{(t)}$, $D^{-}_{(t)}$ and $D^{\pm}_{(t)}$ in Section~\ref{sc:fwer} are designed for controlling the family-wise error rate under a correlated case. However, one may also want to control the false discovery rate (FDR) when desiring large power in practice. We now provide a BH-like procedure for our iterative method to control FDR without proof \citep{benjamini1995controlling}. Specifically, if we choose $D^{+}_{(t)}=D^{-}_{(t)}=T$ and $D^{\pm}_{(t)}=T+r_QM$ for some constant $T$, a plug-in estimator for the upper bound of false discovery proportion is
$$
\widehat{{\rm FDP}}(T)\le {|I_0|\PP\left(Z>T\right)\over \max\{|\hat{I}_1(T)|,1\}},
$$
where $|\hat{I}_1(T)|$ is the total number of rejected null hypotheses by the iterative method and $Z=\sqrt{z_{1}^2+z_{2}^2}$, where $z_{1}$ and $z_{2}$ are independent standard Gaussian random variables. This estimator naturally leads to the following procedure to choose threshold $\hat{T}$
$$
\hat{T}=\inf\left\{0\le T\le q_\alpha: {d\PP\left(Z>T\right)\over \max\{|\hat{I}_1(T)|,1\}\le \alpha}\right\}.
$$
After choosing $\hat{T}$, we can set $D^{+}_{(t)}=D^{-}_{(t)}=\hat{T}$ and $D^{\pm}_{(t)}=\hat{T}+r_QM$ in our iterative method. Simialr to the BH procedure, the guarantee of this procedure may need different components to have independence or some special correlation structure, as a plug-in estimator is used to estimte the false discovery proportion. In practice, when the number of iterations is small, $Z$ can also be chosen as $|z|$, where $z$ is a standard Gaussian random variable. 

%%%%%%%%%%%%%%%%%%%%%%%%%%%%%%%%%%%%%%%%%%%%%%%
\subsection{Continuous Outcome}
%%%%%%%%%%%%%%%%%%%%%%%%%%%%%%%%%%%%%%%%%%%%%%%

Although the discussions in Sections~\ref{sc:model} and \ref{sc:method} focus on two-sample comparison, the idea of this iterative method can also be applied to testing the association between compositional data and a continuous outcome. More concretely, suppose we observe $(\hat{P}_1,Y_1),\ldots, (\hat{P}_m,Y_m)$ where $\hat{P}$ is the compositional data and $Y$ is a continuous variable of interest. If we are interested in testing the linear association between compositional data and outcome, we can replace the two-sample $t$ test with a correlation test 
$$
\hat{R}_i(I)=r_i(I)\sqrt{m-2\over 1-r_i^2(I)},
$$
where $r_i(I)$ is the Pearson correlation coefficient between renormalized compositional data $\hat{P}_{j,i}/\sum_{i\in I}\hat{P}_{j,i}$ and outcome $Y_j$. The theoretical analysis of testing linear association is similar to the two-sample comparison case. However, it remains unclear whether this framework can be extended to testing the nonlinear association between compositional data and continuous outcome.

%%%%%%%%%%%%%%%%%%%%%%%%%%%%%%%%%%%%%%%%%%%%%%%
\section{Numerical Experiments}
\label{sc:nb}
%%%%%%%%%%%%%%%%%%%%%%%%%%%%%%%%%%%%%%%%%%%%%%%

%In this section, we study the numerical performance of our newly proposed method. We carry out simulation studies in Section~\ref{sc:simulation} and real data analysis in Section~\ref{sc:realdata}.

%%%%%%%%%%%%%%%%%%%%%%%%%%
\subsection{Simulation Studies}
\label{sc:simulation}
%%%%%%%%%%%%%%%%%%%%%%%%%%

In simulation experiments, we compare the new robust differential abundance (RDB) test with six other common-used differential abundance tests. The tests we consider here are as follows: analysis of composition of microbiomes (ANCOM) by \cite{mandal2015analysis}, analysis of composition of microbiomes with bias correction (ANCOM.BC) by \cite{lin2020analysis}, differential abundance testing with compositionality adjustment (DACOMP) by \cite{brill2019testing}, differential expression analysis for sequence count data 2 (DESeq2) by \cite{love2014moderated}, and Wilcoxon rank sum test without normalization (Wilcoxon) and with total sum scaling normalization (Wilcoxon.TSS). The simulation experiments consider four different ways to simulate observed data: Poisson-Gamma distributions, log-normal distributions, log-normal distributions with covariates, and real data shuffling. The Supplementary Material contains the detailed settings of simulation experiments.

\begin{table}[h!]
	\small
	\centering
	\begin{tabular}{cccccccc}
		\hline\hline
		&& \multicolumn{2}{c}{$\rho=0.4$} & \multicolumn{2}{c}{$\rho=0.6$}  & \multicolumn{2}{c}{$\rho=0.8$}  \\ 
		&& FWER & Power &  FWER & Power & FWER & Power \\ 
		\hline
		\multirow{7}{*}{S1}&RDB & 0.02 (0.01)&0.57 (0.01)&0.01 (0.01)&0.59 (0.02)&0.01 (0.01)&0.58 (0.01)\\
		&ANCOM.BC & 0.63 (0.05)&0.72 (0.01)&0.60 (0.05)&0.73 (0.01)&0.63 (0.05)&0.71 (0.01)\\
		&ANCOM & 0.17 (0.04)&0.61 (0.01)&0.19 (0.04)&0.63 (0.01)&0.22 (0.04)&0.60 (0.01)\\
		&DESeq2 & 0.98 (0.01)&0.73 (0.01)&0.95 (0.02)&0.74 (0.01)&0.92 (0.03)&0.72 (0.01)\\
		&Wilcoxon & 0.65 (0.05)&0.61 (0.01)&0.58 (0.05)&0.63 (0.01)&0.61 (0.05)&0.60 (0.01)\\
		&Wilcoxon.TSS & 0.70 (0.05)&0.63 (0.01)&0.70 (0.05)&0.65 (0.01)&0.66 (0.05)&0.63 (0.01)\\
		&DACOMP & 0.41 (0.05)&0.67 (0.01)&0.45 (0.05)&0.69 (0.01)&0.38 (0.05)&0.66 (0.01)\\
		\hline
		\multirow{7}{*}{S2}&RDB & 0.01 (0.01)&0.29 (0.01)&0.01 (0.01)&0.30 (0.01)&0.00 (0.00)&0.32 (0.01)\\
		&ANCOM.BC & 0.37 (0.05)&0.43 (0.01)&0.43 (0.05)&0.44 (0.01)&0.46 (0.05)&0.44 (0.01)\\
		&ANCOM & 0.07 (0.03)&0.33 (0.01)&0.08 (0.03)&0.32 (0.01)&0.07 (0.03)&0.34 (0.01)\\
		&DESeq2 & 0.75 (0.04)&0.46 (0.01)&0.73 (0.04)&0.46 (0.01)&0.74 (0.04)&0.46 (0.01)\\
		&Wilcoxon & 0.29 (0.05)&0.41 (0.01)&0.39 (0.05)&0.42 (0.01)&0.39 (0.05)&0.44 (0.01)\\
		&Wilcoxon.TSS & 0.43 (0.05)&0.43 (0.01)&0.44 (0.05)&0.44 (0.01)&0.47 (0.05)&0.47 (0.01)\\
		&DACOMP & 0.21 (0.04)&0.37 (0.01)&0.29 (0.05)&0.38 (0.01)&0.28 (0.05)&0.39 (0.01)\\
		\hline\hline
	\end{tabular}
	\caption{Comparison of differential abundance tests when components are correlated. }
	\label{tb:correlatedcase}
\end{table}

We first simulate the data from Poisson-Gamma distributions and compare these seven different methods when the number of components $d$, number of differentially components $s$, and sample size $m$ vary. The results are summarized in Tables~\ref{tb:numtaxta}, \ref{tb:sparsity}, and \ref{tb:samplesize} when we aim to control FWER at the 10\% level. These results suggest that both ANCOM, ANCOM.BC, DACOMP and RDB tests can control the family-wise error rate relatively well, while other tests suffer from inflated FWER under our simulation settings. An interesting observation is that FWER in RDB, ANCOM, and ANCOM.BC is slightly inflated when the sample size becomes small; however, the FWER in other methods is inflated when the sample size increases. A similar phenomenon is observed and explained in \cite{lin2020analysis}. To compare performances on correlated data, we next simulate the data from log-normal distributions with a correlated covariance matrix and summarize the results in Table~\ref{tb:correlatedcase}. From Table~\ref{tb:correlatedcase}, we can observe that the RDB test can control FWER very well, while other tests inflate the FWER. We also compare the computation time of these seven differential abundance tests when drawing data from Poisson-Gamma distributions. The computation time summarized in Table~\ref{tb:time} suggests that our new test is very computationally efficient. To better understand the RDB test, we also study its performance when the sequencing depth is unbalanced. The results in Table~\ref{tb:unbsdep} suggest that the FWER in RDB is slightly inflated in the unbalanced sequencing depth.

\begin{table}[h!]
	\centering
	\begin{tabular}{ccccc}
		\hline\hline
		&& $d=100$ & $d=200$ & $d=500$ \\ 
		\hline
		RDB && 0.0041 & 0.0081 & 0.0225 \\ 
		ANCOM.BC && 0.7156 & 1.2223 & 2.8303  \\ 
		ANCOM && 15.6485 & 62.5443 & 400.1344  \\ 
		DESeq2 && 1.31384 & 1.3924 & 1.6444  \\ 
		Wilcoxon && 0.0282& 0.0533 & 0.1258  \\ 
		Wilcoxon.TSS && 0.0280 & 0.0552 & 0.1386  \\ 
		DACOMP && 0.4300 & 1.6985 & 8.3371  \\ 
		\hline\hline
	\end{tabular}
	\caption{Comparison of differential abundance tests in computation time (in seconds).}
	\label{tb:time}
\end{table}

We now investigate the performance of the RDB test extensions. Similar to previous experiments, we simulate the data from both Poisson-Gamma and log-normal distributions and compare these seven different methods; however, we now aim to control FDR at the 10\% level. Tables~\ref{tb:samplesizefdr}, \ref{tb:unbsdepfdr}, \ref{tb:correlatedcasefdr}, and \ref{tb:smallnumtaxta} summarize the results. Compared with FWER control, FDR control can boost power for all methods. RDB, ANCOM.BC, ANCOM, and DACOMP can control FDR well in the Poisson-Gamma case. In Table~\ref{tb:correlatedcasefdr}, only RDB can control the FDR in the correlated case; however we remain cautious, as it does not have a theoretical guarantee of controling FDR under arbitrary correlated settings. Next, we compare the performance of different methods when observing some confounding variables. We simulate the data from log-normal distributions with covariates and summarize the results in Table~\ref{tb:cvb}. From Table~\ref{tb:cvb}, the RDB test can control both FWER and FDR after covariate balancing adjustment. Finally, we study the performance of different methods when the outcome is continuous. We draw the data from a Poisson-Gamma distribution with a continuous outcome. The results in Table~\ref{tb:contoutcome} suggest that the RDB test can control FWER well in such Poisson-Gamma distribution.

%%%%%%%%%%%%%%%%%%%%%%%%%%
\subsection{Analysis of Gut Microbiome Data}
\label{sc:realdata}
%%%%%%%%%%%%%%%%%%%%%%%%%%

To further demonstrate the newly proposed method's practical merits, we apply it to a gut microbiota dataset of 476 samples collected in \cite{yatsunenko2012human}. These samples are divided into three groups based on their countries: 83 samples from Malawi (MA), 83 samples from Venezuela (VE), and 310 samples from the USA (US). This data set consists of 11905 OTUs, and we aggregate the dataset at the genus level, which leads to 649 different genera. After aggregation, we have 82.5\% zero entries in the dataset. Besides the countries, each sample in this dataset has two observed covariates: gender (an indicator for female) and age. The goal here is to compare the microbiome communities between different countries and identify differential genera. Before conducting the differential abundance analysis, we first investigate the degree of covariance balancing in this dataset. The violin plot of age is shown in Figure~\ref{fg:realcvb}, and the proportions of the females in Malawi, USA, and Venezuela are 63.9\%, 58.1\%, and 56.6\%, respectively. They suggest that neither age nor gender is well balanced across counties, and thus, covariate adjustment is needed in differential abundance analysis. 

%To balance the covariates, we still apply empirical balancing calibration weighting (CAL) in this section.

\begin{figure}[h!]
	\centering
	\begin{subfigure}[b]{0.4\textwidth}
		\centering
		\includegraphics[width=\textwidth]{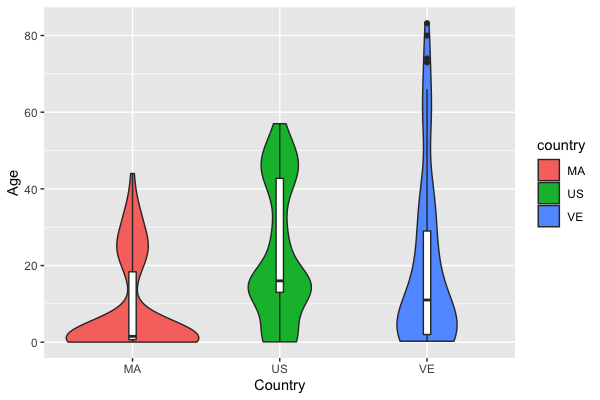}
		\caption{Violin plot of  age}\label{fg:realcvb}
	\end{subfigure}
	\begin{subfigure}[b]{0.52\textwidth}
		\centering
		\includegraphics[width=\textwidth]{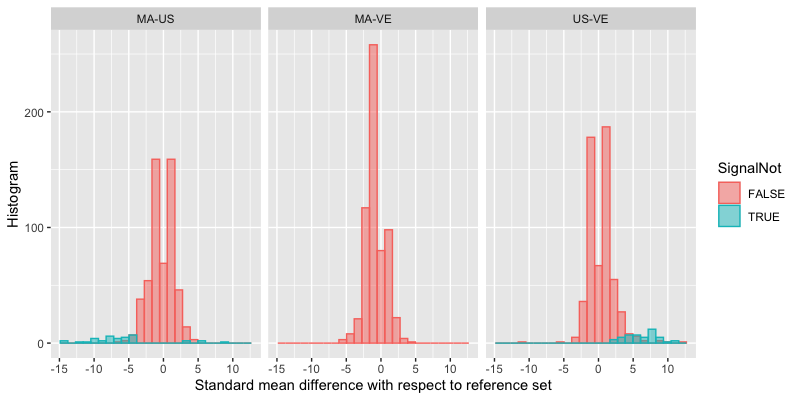}
		\caption{Histograms of standard mean difference}\label{fg:rdbhistogram}
	\end{subfigure}
	\caption{Left figure is a violin plot of age and right figure is the histograms of standard mean difference with respect to the reference set estimated by the RDB test.}\label{fg:covariate}
\end{figure}

We apply the new RDB test to each pair of countries to compare their microbiome communities. The RDB test identifies 37 differential genera between Malawi and USA, no differential genera between Malawi and Venezuela, and 47 between Venezuela and USA. We report the detailed genera in the Supplementary Material. There are 28 common genera reported between Malawi vs. USA and Venezuela vs. USA. In Figure~\ref{fg:rdbhistogram}, we consider all estimated non-differential genera as the reference set in each pairwise comparison, and summarize the standard mean difference with respect to this reference set. Figure~\ref{fg:rdbhistogram} shows that differential and non-differential genera are well separated using this reference set as a benchmark. The results suggest that there is no significant difference between Malawi and Venezuela.  All differential genera between USA and Venezuela have a larger abundance in the USA than Venezuela. Most of the differential genera between USA and Malawi have a larger abundance in the USA than Malawi. Besides, there is large overlap between differential genera in USA vs Venezuela and USA vs Malawi. Among the 28 common genera reported by these two comparisons, 12 belong to order {\it Clostridiales}  and phylum {\it Firmicutes}, which could be explained by a high-fat diet in the USA \citep{clarke2012gut}. Therefore, the pairwise analysis results are consistent with each other. We also apply the RDB test with FDR control, ANCOM, and ANCOM.BC to the same dataset. The results are summarized in Section~\ref{sc:compreal} of the Supplementary Material, suggesting some differences in the reported differential genera. This comparison can also conclude that the RDB test can better identify the differential components from an angle of reference-based hypothesis. 

%%%%%%%%%%%%%%%%%%%%%%%%%%%%%%%%%%%%%%%%%%%%%%%
\section{Concluding Remarks}
\label{sc:conclude}
%%%%%%%%%%%%%%%%%%%%%%%%%%%%%%%%%%%%%%%%%%%%%%%

This study only focuses on the RDB test with a standard two-sample $t$ test for the sake of clean form and theoretical analysis. However, as a general framework, the RDB test can also work with other popular two-sample location tests, such as the Mann-Whitney $U$ test. Using more complicated tests in the RDB test might improve practical performance; however, it will also involve more theoretical analyses. In addition, we mainly consider weighting methods to adjust the confounding effects of observed covariates. As noted in the previous sections, it is possible to apply other covariate balancing techniques too. For example, the RDB test can work with another popular covariate balancing technique, the matching method \citep{rosenbaum2010design,yu2019directional}. Application of RDB test to microbiome data usually requires choosing a suitable taxonomic rank for microbes. The RDB test can also be performed with multi-scale adaptive techniques to detect differential abundant microbes at a high resolution \citep{wang2021multi}. Finally, it could also be interesting to explore if the new framework can test other characteristics of abundances, such as the dispersion and skewness \citep{martin2020modeling}. 

\section*{Acknowledgment}
We thank the editor, associate editor and referees for valuable suggestions. This project is supported by the grant from the National Science Foundation (DMS-2113458).

\bibliographystyle{plainnat}
\bibliography{RDB}

\begin{thebibliography}{41}
\providecommand{\natexlab}[1]{#1}
\providecommand{\url}[1]{\texttt{#1}}
\expandafter\ifx\csname urlstyle\endcsname\relax
  \providecommand{\doi}[1]{doi: #1}\else
  \providecommand{\doi}{doi: \begingroup \urlstyle{rm}\Url}\fi

\bibitem[Aitchison(1983)]{aitchison1983principal}
J.~Aitchison.
\newblock Principal component analysis of compositional data.
\newblock \emph{Biometrika}, 70\penalty0 (1):\penalty0 57--65, 1983.

\bibitem[Benjamini and Hochberg(1995)]{benjamini1995controlling}
Y.~Benjamini and Y.~Hochberg.
\newblock Controlling the false discovery rate: a practical and powerful
  approach to multiple testing.
\newblock \emph{Journal of the Royal statistical society: series B
  (Methodological)}, 57\penalty0 (1):\penalty0 289--300, 1995.

\bibitem[Brill et~al.(2019)Brill, Amir, and Heller]{brill2019testing}
B.~Brill, A.~Amir, and R.~Heller.
\newblock Testing for differential abundance in compositional counts data, with
  application to microbiome studies.
\newblock \emph{arXiv preprint arXiv:1904.08937}, 2019.

\bibitem[Butler et~al.(2018)Butler, Hoffman, Smibert, Papalexi, and
  Satija]{butler2018integrating}
A.~Butler, P.~Hoffman, P.~Smibert, E.~Papalexi, and R.~Satija.
\newblock Integrating single-cell transcriptomic data across different
  conditions, technologies, and species.
\newblock \emph{Nature Biotechnology}, 36\penalty0 (5):\penalty0 411--420,
  2018.

\bibitem[Cao et~al.(2020)Cao, Zhang, and Li]{cao2020multisample}
Y.~Cao, A.~Zhang, and H.~Li.
\newblock Multisample estimation of bacterial composition matrices in
  metagenomics data.
\newblock \emph{Biometrika}, 107\penalty0 (1):\penalty0 75--92, 2020.

\bibitem[Chan et~al.(2016)Chan, Yam, and Zhang]{chan2016globally}
K.~C.~G. Chan, S.~C.~P. Yam, and Z.~Zhang.
\newblock Globally efficient non-parametric inference of average treatment
  effects by empirical balancing calibration weighting.
\newblock \emph{Journal of the Royal Statistical Society: Series B (Statistical
  Methodology)}, 78\penalty0 (3):\penalty0 673--700, 2016.

\bibitem[Clarke et~al.(2012)Clarke, Murphy, Nilaweera, Ross, Shanahan,
  O’Toole, and Cotter]{clarke2012gut}
S.~F. Clarke, E.~F. Murphy, K.~Nilaweera, P.~R. Ross, F.~Shanahan, P.~W.
  O’Toole, and P.~D. Cotter.
\newblock The gut microbiota and its relationship to diet and obesity: new
  insights.
\newblock \emph{Gut Microbes}, 3\penalty0 (3):\penalty0 186--202, 2012.

\bibitem[Delyon(2009)]{delyon2009exponential}
B.~Delyon.
\newblock Exponential inequalities for sums of weakly dependent variables.
\newblock \emph{Electronic Journal of Probability}, 14:\penalty0 752--779,
  2009.

\bibitem[Efron(2004)]{efron2004large}
B.~Efron.
\newblock Large-scale simultaneous hypothesis testing: the choice of a null
  hypothesis.
\newblock \emph{Journal of the American Statistical Association}, 99\penalty0
  (465):\penalty0 96--104, 2004.

\bibitem[Efron(2012)]{efron2012large}
B.~Efron.
\newblock \emph{Large-scale inference: empirical Bayes methods for estimation,
  testing, and prediction}, volume~1.
\newblock Cambridge University Press, 2012.

\bibitem[Fernandes et~al.(2014)Fernandes, Reid, Macklaim, McMurrough, Edgell,
  and Gloor]{fernandes2014unifying}
A.~D. Fernandes, J.~N. Reid, J.~M. Macklaim, T.~A. McMurrough, D.~R. Edgell,
  and G.~B. Gloor.
\newblock Unifying the analysis of high-throughput sequencing datasets:
  characterizing rna-seq, 16s rrna gene sequencing and selective growth
  experiments by compositional data analysis.
\newblock \emph{Microbiome}, 2\penalty0 (1):\penalty0 15, 2014.

\bibitem[Hawinkel et~al.(2019)Hawinkel, Mattiello, Bijnens, and
  Thas]{hawinkel2019broken}
S.~Hawinkel, F.~Mattiello, L.~Bijnens, and O.~Thas.
\newblock A broken promise: microbiome differential abundance methods do not
  control the false discovery rate.
\newblock \emph{Briefings in Bioinformatics}, 20\penalty0 (1):\penalty0
  210--221, 2019.

\bibitem[Imai and Ratkovic(2014)]{imai2014covariate}
K.~Imai and M.~Ratkovic.
\newblock Covariate balancing propensity score.
\newblock \emph{Journal of the Royal Statistical Society: Series B (Statistical
  Methodology)}, 76\penalty0 (1):\penalty0 243--263, 2014.

\bibitem[Imbens and Rubin(2015)]{imbens2015causal}
G.~W. Imbens and D.~B. Rubin.
\newblock \emph{Causal inference in statistics, social, and biomedical
  sciences}.
\newblock Cambridge University Press, 2015.

\bibitem[Kharchenko et~al.(2014)Kharchenko, Silberstein, and
  Scadden]{kharchenko2014bayesian}
P.~V. Kharchenko, L.~Silberstein, and D.~T. Scadden.
\newblock Bayesian approach to single-cell differential expression analysis.
\newblock \emph{Nature methods}, 11\penalty0 (7):\penalty0 740--742, 2014.

\bibitem[Law et~al.(2014)Law, Chen, Shi, and Smyth]{law2014voom}
C.~W Law, Y.~Chen, W.~Shi, and G.~K. Smyth.
\newblock voom: Precision weights unlock linear model analysis tools for
  rna-seq read counts.
\newblock \emph{Genome Biology}, 15\penalty0 (2):\penalty0 R29, 2014.

\bibitem[L{\^e}~Cao et~al.(2016)L{\^e}~Cao, Costello, Lakis, Bartolo, Chua,
  Brazeilles, and Rondeau]{le2016mixmc}
K.~L{\^e}~Cao, M.~Costello, V.~Lakis, F.~Bartolo, X.~Chua, R.~Brazeilles, and
  P.~Rondeau.
\newblock Mixmc: a multivariate statistical framework to gain insight into
  microbial communities.
\newblock \emph{PloS One}, 11\penalty0 (8):\penalty0 e0160169, 2016.

\bibitem[Lin and Peddada(2020)]{lin2020analysis}
H.~Lin and S.~Peddada.
\newblock Analysis of compositions of microbiomes with bias correction.
\newblock \emph{Nature Communications}, 11\penalty0 (1):\penalty0 1--11, 2020.

\bibitem[Love et~al.(2014)Love, Huber, and Anders]{love2014moderated}
M.~I. Love, W.~Huber, and S.~Anders.
\newblock Moderated estimation of fold change and dispersion for rna-seq data
  with deseq2.
\newblock \emph{Genome Biology}, 15\penalty0 (12):\penalty0 550, 2014.

\bibitem[Mandal et~al.(2015)Mandal, Van~Treuren, White, Eggesb{\o}, Knight, and
  Peddada]{mandal2015analysis}
S.~Mandal, W.~Van~Treuren, R.~A. White, M.~Eggesb{\o}, R.~Knight, and S.~D.
  Peddada.
\newblock Analysis of composition of microbiomes: a novel method for studying
  microbial composition.
\newblock \emph{Microbial ecology in health and disease}, 26\penalty0
  (1):\penalty0 27663, 2015.

\bibitem[Martin et~al.(2020)Martin, Witten, and Willis]{martin2020modeling}
B.~D. Martin, D.~Witten, and A.~D. Willis.
\newblock Modeling microbial abundances and dysbiosis with beta-binomial
  regression.
\newblock \emph{Annals of Applied Statistics}, 14\penalty0 (1):\penalty0
  94--115, 2020.

\bibitem[Morton et~al.(2019)Morton, Marotz, Washburne, Silverman, Zaramela,
  Edlund, Zengler, and Knight]{morton2019establishing}
J.~T. Morton, C.~Marotz, A.~Washburne, J.~Silverman, L.~S. Zaramela, A.~Edlund,
  K.~Zengler, and R.~Knight.
\newblock Establishing microbial composition measurement standards with
  reference frames.
\newblock \emph{Nature Communications}, 10\penalty0 (1):\penalty0 1--11, 2019.

\bibitem[Paulson et~al.(2013)Paulson, Stine, Bravo, and
  Pop]{paulson2013differential}
J.~N. Paulson, O.~C. Stine, H.~C. Bravo, and M.~Pop.
\newblock Differential abundance analysis for microbial marker-gene surveys.
\newblock \emph{Nature Methods}, 10\penalty0 (12):\penalty0 1200--1202, 2013.

\bibitem[Pawlowsky-Glahn and Buccianti(2011)]{pawlowsky2011compositional}
V.~Pawlowsky-Glahn and A.~Buccianti.
\newblock \emph{Compositional data analysis: Theory and applications}.
\newblock John Wiley \& Sons, 2011.

\bibitem[Risso et~al.(2018)Risso, Perraudeau, Gribkova, Dudoit, and
  Vert]{risso2018general}
D.~Risso, F.~Perraudeau, S.~Gribkova, S.~Dudoit, and J.~Vert.
\newblock A general and flexible method for signal extraction from single-cell
  rna-seq data.
\newblock \emph{Nature communications}, 9\penalty0 (1):\penalty0 1--17, 2018.

\bibitem[Robbins(1951)]{robbins1951asymptotically}
H.~Robbins.
\newblock Asymptotically subminimax solutions of compound statistical decision
  problems.
\newblock In \emph{Proceedings of the Second Berkeley Symposium on Mathematical
  Statistics and Probability}. The Regents of the University of California,
  1951.

\bibitem[Robins et~al.(2000)Robins, Hernan, and Brumback]{robins2000marginal}
J.~M. Robins, M.~A. Hernan, and B.~Brumback.
\newblock Marginal structural models and causal inference in epidemiology.
\newblock \emph{Epidemiology}, 11\penalty0 (5):\penalty0 550--560, 2000.

\bibitem[Robinson et~al.(2010)Robinson, McCarthy, and Smyth]{robinson2010edger}
M.~D. Robinson, D.~J. McCarthy, and G.~K. Smyth.
\newblock edger: a bioconductor package for differential expression analysis of
  digital gene expression data.
\newblock \emph{Bioinformatics}, 26\penalty0 (1):\penalty0 139--140, 2010.

\bibitem[Rosenbaum(1987)]{rosenbaum1987model}
P.~R. Rosenbaum.
\newblock Model-based direct adjustment.
\newblock \emph{Journal of the American Statistical Association}, 82\penalty0
  (398):\penalty0 387--394, 1987.

\bibitem[Rosenbaum(2002)]{rosenbaum2010design}
P.~R. Rosenbaum.
\newblock \emph{Design of observational studies}.
\newblock Springer, 2002.

\bibitem[Rosenbaum and Rubin(1983)]{rosenbaum1983central}
P.~R. Rosenbaum and D.~B. Rubin.
\newblock The central role of the propensity score in observational studies for
  causal effects.
\newblock \emph{Biometrika}, 70\penalty0 (1):\penalty0 41--55, 1983.

\bibitem[Talagrand(2014)]{talagrand2014upper}
M.~Talagrand.
\newblock \emph{Upper and lower bounds for stochastic processes: modern methods
  and classical problems}, volume~60.
\newblock Springer Science \& Business Media, 2014.

\bibitem[Vandeputte et~al.(2017)Vandeputte, Kathagen, D’hoe, Vieira-Silva,
  Valles-Colomer, Sabino, Wang, Tito, De~Commer, Darzi, Vermeire, Falony, and
  Raes]{vandeputte2017quantitative}
D.~Vandeputte, G.~Kathagen, K.~D’hoe, S.~Vieira-Silva, M.~Valles-Colomer,
  J.~Sabino, J.~Wang, R.~Y. Tito, L.~De~Commer, Y.~Darzi, S.~Vermeire,
  G.~Falony, and J.~Raes.
\newblock Quantitative microbiome profiling links gut community variation to
  microbial load.
\newblock \emph{Nature}, 551\penalty0 (7681):\penalty0 507--511, 2017.

\bibitem[Wang(2021)]{wang2021multi}
S.~Wang.
\newblock Multi-scale adaptive differential abundance analysis in microbial
  compositional data.
\newblock \emph{bioRxiv}, 2021.

\bibitem[Wang et~al.(2020)Wang, Cai, and Li]{wang2020hypothesis}
S.~Wang, T.~T. Cai, and H.~Li.
\newblock Hypothesis testing for phylogenetic composition: a minimum-cost flow
  perspective.
\newblock \emph{Biometrika}, 2020.

\bibitem[Wang et~al.(2021)Wang, Fan, Pocock, Arena, Eliceiri, and
  Yuan]{wang2019structured}
S.~Wang, J.~Fan, G.~Pocock, E.~T. Arena, K.~W. Eliceiri, and M.~Yuan.
\newblock Structured correlation detection with application to colocalization
  analysis in dual-channel fluorescence microscopic imaging.
\newblock \emph{Statistica Sinica}, 31\penalty0 (1):\penalty0 333--360, 2021.

\bibitem[Weiss et~al.(2017)Weiss, Xu, Peddada, Amir, Bittinger, Gonzalez,
  Lozupone, Zaneveld, V{\'a}zquez-Baeza, Birmingham,
  et~al.]{weiss2017normalization}
S.~Weiss, Z.~Z. Xu, S.~Peddada, A.~Amir, K.~Bittinger, A.~Gonzalez,
  C.~Lozupone, J.~R Zaneveld, Y.~V{\'a}zquez-Baeza, A.~Birmingham, et~al.
\newblock Normalization and microbial differential abundance strategies depend
  upon data characteristics.
\newblock \emph{Microbiome}, 5\penalty0 (1):\penalty0 27, 2017.

\bibitem[Yatsunenko et~al.(2012)Yatsunenko, Rey, Manary, Trehan,
  Dominguez-Bello, Contreras, Magris, Hidalgo, Baldassano, Anokhin, Heath,
  Warner, Reeder, Kuczynski, Caporaso, Lozupone, Lauber, Clemente, Knights,
  Knight, and Gordon]{yatsunenko2012human}
T.~Yatsunenko, F.~E. Rey, M.~J. Manary, I.~Trehan, M.~G. Dominguez-Bello,
  M.~Contreras, M.~Magris, G.~Hidalgo, R.~N. Baldassano, A.~P. Anokhin, A.~C.
  Heath, B.~Warner, J.~Reeder, J.~Kuczynski, J.~G. Caporaso, C.~A. Lozupone,
  C.~Lauber, J.~C. Clemente, D.~Knights, R.~Knight, and J.~I. Gordon.
\newblock Human gut microbiome viewed across age and geography.
\newblock \emph{Nature}, 486\penalty0 (7402):\penalty0 222--227, 2012.

\bibitem[Yu and Rosenbaum(2019)]{yu2019directional}
R.~Yu and P.~R. Rosenbaum.
\newblock Directional penalties for optimal matching in observational studies.
\newblock \emph{Biometrics}, 75\penalty0 (4):\penalty0 1380--1390, 2019.

\bibitem[Yu and Wang(2020)]{yu2020covariate}
R.~Yu and S.~Wang.
\newblock Covariate balancing by uniform transformer.
\newblock \emph{arXiv preprint arXiv:2008.03738}, 2020.

\bibitem[Zaitsev(1987)]{zaitsev1987gaussian}
A.~Y. Zaitsev.
\newblock On the gaussian approximation of convolutions under multidimensional
  analogues of sn bernstein's inequality conditions.
\newblock \emph{Probability Theory and Related Fields}, 74\penalty0
  (4):\penalty0 535--566, 1987.

\end{thebibliography}

%%%%%%%%%%%%%%%%%%%%%%%%%%%%%%%%%%%%%%%%%%%%%%%
\section*{Appendix A}
%%%%%%%%%%%%%%%%%%%%%%%%%%%%%%%%%%%%%%%%%%%%%%%

In this appendix, we provide some extra results on numerical experiments. 

%%%%%%%%%%%%%%%%%%%%%%%%%%
\subsection{Simulation Settings}
%%%%%%%%%%%%%%%%%%%%%%%%%%

The simulation experiments consider four ways to simulate observed data: Poisson-Gamma distributions, log-normal distributions, log-normal distributions with covariates and real data shuffling. We draw $m_1$ samples from the first population and $m_2$ samples from the second population for two-group comparisons. There are total $d$ components and $s$ differential components. The effect size of differential ones can be represented as a vector $a=(a_1,a_2,\ldots,a_d)$. $s$ components are randomly drawn, so $a_i=1$ if $i$ is a non-differential component. There are two settings for $s$ differential components: (i) we draw the effect size of differential components from a uniform distribution between 1 and 5; (ii) we draw half of the differential components' effect size from a uniform distribution between 1 and 5, and the other half from a uniform distribution between 0.2 and 1. If the sequencing depth is balanced between two groups, we draw the total number of counts $N^\ast_{k,j}$ randomly between 5000 and 50000 for each sample. If not, we draw $N^\ast_{k,j}$ randomly between 5000 and 50000 when $k=1$ and between $5000/\beta$ and $50000/\beta$ when $k=2$ for some $\beta\ge 1$. 

\begin{enumerate}[label=(\alph*)]
	\item {\it Poisson-Gamma Distributions.} In Poisson-Gamma distributions, we first randomly generate a Gamma distribution parameter vector $\gamma=(\gamma_1,\gamma_2,\ldots,\gamma_d)$ so that 60\% of them are 50, 30\% are 200, and 10\% are 10000. Then, given $\gamma$ and effect size $a$, the absolute abundance $A_{k,j,i}\sim {\rm Pois}(\gamma_i)$ if $k=1$ and $A_{k,j,i}\sim {\rm Pois}(a_i\gamma_i)$ if $k=2$ in a two-sample comparison case. In the continuous outcome case, we first draw each $Y_j$ from a uniform distribution between 0 and 1 and then draw the absolute abundance $A_{j,i}\sim {\rm Pois}[\{1+(a_i-1)Y_j\}\gamma_i]$. Finally, given the absolute abundance, we draw the count data from the model in Section~\ref{sc:model}.  
	\item {\it Log-Normal Distributions.} In log-normal distributions, we generate a mean vector $\mu=(\mu_1,\mu_2,\ldots,\mu_d)$ so that 60\% of them are 3, 30\% are 5, and 10\% are 10. Then, we set the covariance matrix $\Sigma$ such that $\Sigma_{l_1,l_2}=\rho^{|l_1-l_2|}$ for some given $\rho<1$. Given $\mu$, $\Sigma$, and effect size $a$, we draw the absolute abundance in the following way: $\log(A_{k,j})\sim N(\mu,\Sigma)$ when $k=1$ and $\log(A_{k,j})\sim N(\log(a)+\mu,\Sigma)$ when $k=2$. Here, $\log$ is an entry-wise operation for a vector. Finally, given the absolute abundance, we draw the count data from the model in Section~\ref{sc:model}. 
	\item {\it Log-Normal Distributions with Covariates.} In log-normal distributions with covariates, the latent covariates for each sample are generated in the following way: $W_{k,j}\sim N(\eta, I_{5\times 5})$ if $k=1$ and $W_{k,j}\sim N(-\eta,I_{5\times 5})$ if $k=2$, where $N$ represents the normal distribution, $\eta=(0.25,\ldots,0.25)\in \RR^5$ and $I_{5\times 5}$ is the $5\times 5$ identity matrix. Instead of observing $W_{k,j}$, we observe $X_{k,j}$ such that $X_{k,j}=\exp(W_{k,j})+W_{k,j}$, where $\exp$ is an entry-wise operation for a vector. Then, we generate an amplitude vector $\lambda=(\lambda_1,\lambda_2,\ldots,\lambda_d)$ so that 60\% of them are 1, 30\% are 2, and 10\% are 3. For the $i$th component, we draw the coefficient $\bb_i\in \RR^5$ in the following way: $\PP(b_i=\lambda_iu_i)=\PP(b_i=-\lambda_iu_i)=1/2$ where each component of $u_i$ follows a uniform distribution between $0$ and $1$. Given $B=(b_1,b_2,\ldots,b_d)$ and $W_{k,j}$, we draw the absolute abundance in the following way: $A_{k,j}=\exp(W_{k,j}^TB)+Z_{k,j}$ when $k=1$ and $A_{k,j}=2a\times\{\exp(W_{k,j}^TB)+Z_{k,j}\}$ when $k=2$. Here, each component of $Z_{k,j}$ is drawn from an exponential distribution. Finally, given the absolute abundance, we draw count data from the model in Section~\ref{sc:model}. 
	\item {\it Real Data Shuffling.} In real data shuffling, we need a count dataset of $m'$ samples $(N'_{1},\ldots,N'_{m'})$ from some real studies. This paper uses the gut microbiome of 476 samples collected in \cite{yatsunenko2012human}. We randomly draw $m_1+m_2$ samples from this dataset without replacement when $m_1+m_2\le 476$. Then, we randomly split these samples into two groups so that there are $m_1$ samples in the first group and $m_2$ samples in the second group. We can think the data are generated from the same distribution by doing this. 
\end{enumerate}

%%%%%%%%%%%%%%%%%%%%%%%%%%
\subsection{Performance Measures}
%%%%%%%%%%%%%%%%%%%%%%%%%%

To evaluate different methods' performance, we consider three different measures: 
\begin{enumerate}[label=(\alph*)]
	\item FWER, which is defined as the probability of having a false discovery, that is, $\PP(\hat{I}_1\cap I_0\ne \emptyset)$.
	\item FDR, which is defined as the expectation of false discovery proportion, that is, $\EE\{|\hat{I}_1\cap I_0|/\max(|\hat{I}_1|,1)\}$.
	\item Power, which is defined as the expected true discovery proportion, that is, $\EE(|\hat{I}_1\cap I_1|/|I_1|)$.
\end{enumerate}

In all simulation experiments, we also report standard error in parenthesis when we report the estimated measures. 

%%%%%%%%%%%%%%%%%%%%%%%%%%
\subsection{Simulation Results}
%%%%%%%%%%%%%%%%%%%%%%%%%%

In this section, we carry out simulation experiments. 

{\it Number of Components.} In this simulation experiment, we aim to compare the performance of different methods when the number of components is different. In particular, we vary $d=200, 500$ and choose Poisson-Gamma distributions with $m_1=m_2=50$, $s=20$ and $\beta=1$. We consider both setting 1 and setting 2 in effect size. This simulation experiment is designed to control FWER at the $10\%$ level. The results based on 100 times simulation experiments are summarized in Table~\ref{tb:numtaxta}. The results in Table~\ref{tb:numtaxta} show that both ANCOM and RDB tests can control FWER very well and ANCOM.BC and DACOMP lead to a little inflated FWER, while other tests suffer highly inflated FWER under this simulation setting. All tests achieve similar power, and there is a bit of power loss for ANCOM, Wilcoxon, and DACOMP.

{\it Number of Differential Components.} In this simulation experiment, we investigate the tests' performance when the number of differential components $s$ is different. We still choose Poisson-Gamma distributions with setting 1 and setting 2 in effect size and set $s=10, 40$, $d=200$, $m_1=m_2=50$, and $\beta=1$. This simulation experiment is designed to control FWER at the $10\%$ level and repeated 100 times. Table~\ref{tb:sparsity} summarizes the results of this simulation experiment. Table~\ref{tb:sparsity} can also be compared with the case of $s=20, d=200$ in Table~\ref{tb:numtaxta}. All differential abundance test performs similarly when the number of differential components is different.

\begin{table}[h!]
	\centering
	\begin{tabular}{cccccc}
		\hline\hline
		&& \multicolumn{2}{c}{$d=200$} & \multicolumn{2}{c}{$d=500$}  \\ 
		&& FWER & Power &  FWER & Power\\ 
		\hline
		\multirow{7}{*}{Setting 1}&RDB &0.04 (0.02)&0.91 (0.01)&0.09 (0.03)&0.85 (0.01)\\
		&ANCOM.BC & 0.19 (0.04)&0.93 (0.01)&0.24 (0.04)&0.88 (0.01)\\
		&ANCOM & 0.00 (0.00)&0.86 (0.01)&0.00 (0.00)&0.78 (0.01)\\
		&DESeq2 & 1.00 (0.00)&0.92 (0.01)&0.19 (0.04)&0.87 (0.01)\\
		&Wilcoxon & 0.27 (0.04)&0.77 (0.01)&0.10 (0.03)&0.77 (0.01)\\
		&Wilcoxon.TSS & 0.99 (0.01)&0.87 (0.01)&0.80 (0.04)&0.84 (0.01)\\
		&DACOMP & 0.13 (0.03)&0.83 (0.01)&0.12 (0.03)&0.70 (0.01)\\
		\hline
		\multirow{7}{*}{Setting 2}&RDB &0.08 (0.03)&0.79 (0.01)&0.04 (0.02)&0.67 (0.01)\\
		&ANCOM.BC & 0.14 (0.03)&0.81 (0.01)&0.17 (0.04)&0.71 (0.01)\\
		&ANCOM &0.00 (0.00)&0.70 (0.01)&0.00 (0.00)&0.55 (0.01)\\
		&DESeq2 &0.13 (0.03)&0.83 (0.01)&0.01 (0.01)&0.69 (0.01)\\
		&Wilcoxon & 0.10 (0.03)&0.71 (0.01)&0.03 (0.02)&0.61 (0.01)\\
		&Wilcoxon.TSS & 0.83 (0.04)&0.84 (0.01)&0.65 (0.05)&0.70 (0.01)\\
		&DACOMP & 0.09 (0.03)&0.60 (0.01)&0.08 (0.03)&0.46 (0.01)\\
		\hline\hline
	\end{tabular}
	\caption{Comparisons of differential abundance tests for different number of components $d$. Standard error is reported in parenthesis.}
	\label{tb:numtaxta}
\end{table}

\begin{table}[h!]
	\centering
	\begin{tabular}{cccccc}
		\hline\hline
		&& \multicolumn{2}{c}{$s=10$} & \multicolumn{2}{c}{$s=40$}  \\ 
		&& FWER & Power &  FWER & Power\\ 
		\hline
		\multirow{7}{*}{Setting 1}&RDB &0.05 (0.02)&0.92 (0.01)&0.02 (0.01)&0.91 (0.00)\\
		&ANCOM.BC & 0.18 (0.04)&0.93 (0.01)&0.48 (0.05)&0.93 (0.00)\\
		&ANCOM & 0.00 (0.00)&0.86 (0.01)&0.00 (0.00)&0.85 (0.01)\\
		&DESeq2 & 0.61 (0.05)&0.91 (0.01)&1.00 (0.00)&0.90 (0.00)\\
		&Wilcoxon & 0.03 (0.02)&0.81 (0.01)&0.68 (0.05)&0.67 (0.01)\\
		&Wilcoxon.TSS & 0.71 (0.05)&0.90 (0.01)&1.00 (0.00)&0.83 (0.01)\\
		&DACOMP & 0.14 (0.03)&0.82 (0.01)&0.30 (0.05)&0.81 (0.01)\\
		\hline
		\multirow{7}{*}{Setting 2}&RDB &0.06 (0.02)&0.78 (0.01)&0.10 (0.03)&0.79 (0.01)\\
		&ANCOM.BC & 0.18 (0.04)&0.80 (0.01)&0.17 (0.04)&0.80 (0.01)\\
		&ANCOM &0.00 (0.00)&0.70 (0.01)&0.00 (0.00)&0.69 (0.01)\\
		&DESeq2 &0.09 (0.03)&0.81 (0.01)&0.55 (0.05)&0.82 (0.01)\\
		&Wilcoxon & 0.03 (0.02)&0.68 (0.01)&0.23 (0.04)&0.70 (0.01)\\
		&Wilcoxon.TSS & 0.53 (0.05)&0.81 (0.01)&0.92 (0.03)&0.85 (0.01)\\
		&DACOMP & 0.09 (0.03)&0.63 (0.01)&0.09 (0.03)&0.59 (0.01)\\
		\hline\hline
	\end{tabular}
	\caption{Comparisons of different differential abundance tests when the number of differential components $s$ is different. Standard error is reported in parenthesis.}
	\label{tb:sparsity}
\end{table}

{\it Sample Size.} In this simulation experiment, we compare the effect of sample size on different differential abundance tests. In particular, we consider two cases: $m_1=m_2=20$ and $m_1=100, m_2=50$. The sample size is unbalanced in the second case. Like the last simulation, we also adopt Poisson-Gamma distributions with both setting 1 and setting 2 in effect size and set $s=20$, $d=200$, and $\beta=1$. We aim to control FWER at the $10\%$ level in this simulation experiment and repeat each case 100 times. In Table~\ref{tb:samplesize}, we summarize the results, which show that the FWER is a bit inflated, and the power is reduced when the sample size is small. 

\begin{table}[h!]
	\centering
	\begin{tabular}{cccccc}
		\hline\hline
		&& \multicolumn{2}{c}{$m=20$} & \multicolumn{2}{c}{$m=100$}  \\ 
		&& FWER & Power &  FWER & Power\\ 
		\hline
		\multirow{7}{*}{Setting 1}&RDB &0.11 (0.03)&0.85 (0.01)&0.04 (0.02)&0.92 (0.01)\\
		&ANCOM.BC & 0.30 (0.05)&0.88 (0.01)&0.25 (0.04)&0.94 (0.01)\\
		&ANCOM & 0.00 (0.00)&0.71 (0.01)&0.00 (0.00)&0.88 (0.01)\\
		&DESeq2 & 0.81 (0.04)&0.86 (0.01)&0.98 (0.01)&0.92 (0.01)\\
		&Wilcoxon & 0.08 (0.03)&0.54 (0.02)&0.31 (0.05)&0.81 (0.01)\\
		&Wilcoxon.TSS & 0.91 (0.03)&0.81 (0.01)&1.00 (0.00)&0.88 (0.01)\\
		&DACOMP & 0.06 (0.02)&0.70 (0.01)&0.13 (0.03)&0.85 (0.01)\\
		\hline
		\multirow{7}{*}{Setting 2}&RDB &0.11 (0.03)&0.69 (0.01)&0.07 (0.03)&0.83 (0.01)\\
		&ANCOM.BC & 0.24 (0.04)&0.74 (0.01)&0.20 (0.04)&0.84 (0.01)\\
		&ANCOM &0.00 (0.00)&0.50 (0.01)&0.00 (0.00)&0.76 (0.01)\\
		&DESeq2 &0.05 (0.02)&0.74 (0.01)&0.18 (0.04)&0.85 (0.01)\\
		&Wilcoxon & 0.03 (0.02)&0.46 (0.01)&0.12 (0.03)&0.76 (0.01)\\
		&Wilcoxon.TSS &0.78 (0.04)&0.72 (0.01)&0.86 (0.03)&0.88 (0.01)\\
		&DACOMP & 0.06 (0.02)&0.47 (0.01)&0.12 (0.03)&0.67 (0.01)\\
		\hline\hline
	\end{tabular}
	\caption{Comparisons of different differential abundance tests when sample size is different. Standard error is reported in parenthesis.}
	\label{tb:samplesize}
\end{table}

{\it Correlation between Components.} In this simulation experiment, we study the performance of different tests when the components are correlated with each other. To simulate the correlated data, we adopt log-normal distributions in this simulation experiment. We choose $\rho=0.4, 0.6$ and $0.8$ for different levels of correlation and set $m_1=m_2=50$, $s=20$, $d=200$, and $\beta=1$. We repeat each case 100 times and aim to control FWER at the $10\%$ level. The results are summarized in Table~\ref{tb:correlatedcase}. From Table~\ref{tb:correlatedcase}, we can observe that the RDB test can control FWER very well while other tests inflate the FWER. In the correlated case, the RDB test is not the most powerful compared with other tests.

{\it Computation Time.} This simulation experiment compares these differential abundance tests' computation time in the same setting as Table~\ref{tb:numtaxta}. We record the computation time from 10  replications with $d=100$, $200$, and $500$, and $m_1=m_2=50$, and all these methods are evaluated on the same desktop  (Intel Core i7 @3.8 GHz/16GB). The experiment results are summarized in Table~\ref{tb:time}. From Table~\ref{tb:time}, we can conclude that the new robust differential abundance (RDB) test is very computationally efficient and can report the results in less than 1 second. 

{\it Unbalanced Sequence Depth.} In this simulation experiment, we study the effect of unbalanced sequencing depth on the RDB test. We set $\beta=1,2,$ and $4$ to vary the level of unbalance. Poisson-Gamma distribution is used in this simulation experiment and $s=20$, $d=200$, $m_1=m_2=50$. Like previous simulation experiments, FWER targets at the $10\%$ level ,and each case repeats 100 times. The results are summarized in Table~\ref{tb:unbsdep}. The results in Table~\ref{tb:unbsdep} suggest that the FWER is inflated a bit, and power is reduced when the sequencing depth becomes unbalanced. Overall the performance is acceptable. This problem could be alleviated if the rank-based test replaces the $t$ test or we adopt sequence rarefication.

\begin{table}[h!]
	\small
	\centering
	\begin{tabular}{cccccccc}
		\hline\hline
		&& \multicolumn{2}{c}{$\beta=1$} & \multicolumn{2}{c}{$\beta=2$}  & \multicolumn{2}{c}{$\beta=4$}  \\ 
		&& FWER & Power &  FWER & Power & FWER & Power \\ 
		\hline
		\multirow{7}{*}{S1}&RDB & 0.04 (0.02)&0.91 (0.01)&0.06 (0.02)&0.88 (0.01)&0.14 (0.03)&0.84 (0.01)\\
		&ANCOM.BC &0.19 (0.04)&0.93 (0.01)&0.99 (0.01)&0.89 (0.01)&1.00 (0.00)&0.82 (0.01)\\
		&ANCOM & 0.00 (0.00)&0.86 (0.01)&0.03 (0.02)&0.79 (0.01)&0.97 (0.02)&0.66 (0.01)\\
		&DESeq2 & 1.00 (0.00)&0.92 (0.01)&1.00 (0.00)&0.91 (0.01)&1.00 (0.00)&0.90 (0.01)\\
		&Wilcoxon &0.27 (0.04)&0.77 (0.01)&1.00 (0.00)&0.44 (0.01)&1.00 (0.00)&0.47 (0.02)\\
		&Wilcoxon.TSS &0.99 (0.01)&0.87 (0.01)&0.99 (0.01)&0.84 (0.01)&0.99 (0.01)&0.79 (0.01)\\
		&DACOMP & 0.13 (0.03)&0.83 (0.01)&0.62 (0.05)&0.73 (0.01)&0.99 (0.01)&0.56 (0.02)\\
		\hline
		\multirow{7}{*}{S2}&RDB & 0.08 (0.03)&0.79 (0.01)&0.12 (0.03)&0.74 (0.01)&0.09 (0.03)&0.71 (0.01)\\
		&ANCOM.BC & 0.14 (0.03)&0.81 (0.01)&0.98 (0.01)&0.75 (0.01)&1.00 (0.00)&0.71 (0.01)\\
		&ANCOM & 0.00 (0.00)&0.70 (0.01)&0.00 (0.00)&0.57 (0.01)&0.90 (0.03)&0.49 (0.01)\\
		&DESeq2 & 0.13 (0.03)&0.83 (0.01)&0.99 (0.01)&0.74 (0.01)&1.00 (0.00)&0.66 (0.01)\\
		&Wilcoxon & 0.10 (0.03)&0.71 (0.01)&1.00 (0.00)&0.74 (0.01)&1.00 (0.00)&0.71 (0.01)\\
		&Wilcoxon.TSS &0.83 (0.04)&0.84 (0.01)&0.83 (0.04)&0.80 (0.01)&0.89 (0.03)&0.79 (0.01)\\
		&DACOMP & 0.09 (0.03)&0.60 (0.01)&0.51 (0.05)&0.49 (0.01)&0.95 (0.02)&0.46 (0.01)\\
		\hline\hline
	\end{tabular}
	\caption{Performance of RDB test when sequencing depth is unbalanced. Standard error is reported in parenthesis. S1=Setting 1 and S2=Setting 2.}
	\label{tb:unbsdep}
\end{table}

{\it False Discovery Rate Control.} In this simulation experiment, we mainly focus on the false discovery rate control. To control the false discovery rate, RDB chooses the critical values based on the method introduced in Section~\ref{sc:cb}, and all other methods adopt the BH procedure to control FDR. We repeat the same simulation experiments in Tables~\ref{tb:samplesize}, \ref{tb:unbsdep}, and \ref{tb:correlatedcase}, but aim to control FDR at the 10\% level. The results are summarized in Tables~\ref{tb:samplesizefdr}, \ref{tb:unbsdepfdr}, and \ref{tb:correlatedcasefdr}, respectively. These results suggest that power increases in all methods, and RDB, ANCOM.BC, ANCOM, and DACOMP can well control FDR in the Poisson-Gamma case. Only RDB can control the FDR in the correlated case, but we shall be cautious about applying RDB to control FDR in practice as it does not have a theoretical guarantee of controling FDR under arbitrary correlated settings. In addition, we also design a simulation experiment to compare the performances of differential abundance tests when the number of components $d$ are relative small. Specifically, we choose $d=30$ and $60$ in Poisson-Gamma distribution with setting 1 and 2 ($m_1=m_2=50$ and $\beta=1$) and summarize the results in Table~\ref{tb:smallnumtaxta}. Since the RDB test is mainly designed for high dimensional problem, the FDR is a bit inflated. 

\begin{table}[h!]
	\centering
	\begin{tabular}{cccccc}
		\hline\hline
		&& \multicolumn{2}{c}{$m=20$} & \multicolumn{2}{c}{$m=100$}  \\ 
		&& FDR & Power &  FDR & Power\\ 
		\hline
		\multirow{7}{*}{Setting 1}&RDB &0.14 (0.01)&0.91 (0.01)&0.11 (0.01)&0.95 (0.00)\\
		&ANCOM.BC & 0.18 (0.01)&0.92 (0.01)&0.23 (0.02)&0.96 (0.00)\\
		&ANCOM & 0.02 (0.00)&0.88 (0.01)&0.03 (0.00)&0.94 (0.01)\\
		&DESeq2 & 0.49 (0.01)&0.91 (0.01)&0.60 (0.00)&0.95 (0.01)\\
		&Wilcoxon & 0.20 (0.03)&0.72 (0.01)&0.49 (0.04)&0.88 (0.01)\\
		&Wilcoxon.TSS & 0.71 (0.02)&0.88 (0.01)&0.80 (0.02)&0.94 (0.01)\\
		&DACOMP & 0.10 (0.01)&0.80 (0.01)&0.12 (0.01)&0.90 (0.01)\\
		\hline
		\multirow{7}{*}{Setting 2}&RDB &0.13 (0.01)&0.79 (0.01)&0.13 (0.01)&0.88 (0.01)\\
		&ANCOM.BC & 0.15 (0.01)&0.80 (0.01)&0.15 (0.02)&0.88 (0.01)\\
		&ANCOM &0.02 (0.00)&0.76 (0.01)&0.03 (0.00)&0.86 (0.01)\\
		&DESeq2 &0.08 (0.01)&0.80 (0.01)&0.17 (0.02)&0.90 (0.01)\\
		&Wilcoxon & 0.08 (0.02)&0.62 (0.01)&0.22 (0.03)&0.83 (0.01)\\
		&Wilcoxon.TSS &0.51 (0.03)&0.83 (0.01)&0.65 (0.02)&0.93 (0.01)\\
		&DACOMP & 0.09 (0.01)&0.60 (0.01)&0.10 (0.01)&0.74 (0.01)\\
		\hline\hline
	\end{tabular}
	\caption{Comparisons of different differential abundance tests when sample size is different and the target is to control FDR. Standard error is reported in parenthesis.}
	\label{tb:samplesizefdr}
\end{table}

\begin{table}[h!]
	\small
	\centering
	\begin{tabular}{cccccccc}
		\hline\hline
		&& \multicolumn{2}{c}{$\beta=1$} & \multicolumn{2}{c}{$\beta=2$}  & \multicolumn{2}{c}{$\beta=4$}  \\ 
		&& FDR & Power &  FDR & Power & FDR & Power \\ 
		\hline
		\multirow{7}{*}{S1}&RDB & 0.10 (0.01)&0.95 (0.01)&0.10 (0.01)&0.93 (0.01)&0.12 (0.01)&0.91 (0.01)\\
		&ANCOM.BC & 0.16 (0.02)&0.96 (0.01)&0.79 (0.01)&0.94 (0.01)&0.82 (0.00)&0.90 (0.01)\\
		&ANCOM & 0.03 (0.00)&0.94 (0.01)&0.44 (0.02)&0.89 (0.01)&0.79 (0.00)&0.86 (0.01)\\
		&DESeq2 & 0.58 (0.00)&0.95 (0.01)&0.63 (0.00)&0.95 (0.01)&0.77 (0.00)&0.96 (0.00)\\
		&Wilcoxon & 0.30 (0.04)&0.87 (0.01)&0.92 (0.00)&0.73 (0.01)&0.92 (0.00)&0.73 (0.01)\\
		&Wilcoxon.TSS & 0.76 (0.02)&0.93 (0.01)&0.77 (0.02)&0.92 (0.01)&0.77 (0.02)&0.90 (0.01)\\
		&DACOMP & 0.09 (0.01)&0.89 (0.01)&0.50 (0.03)&0.83 (0.01)&0.76 (0.01)&0.73 (0.01)\\
		\hline
		\multirow{7}{*}{S2}&RDB & 0.12 (0.01)&0.86 (0.01)&0.13 (0.01)&0.84 (0.01)&0.14 (0.01)&0.80 (0.01)\\
		&ANCOM.BC & 0.13 (0.01)&0.86 (0.01)&0.77 (0.01)&0.86 (0.01)&0.84 (0.00)&0.83 (0.01)\\
		&ANCOM & 0.03 (0.00)&0.84 (0.01)&0.27 (0.02)&0.75 (0.01)&0.78 (0.01)&0.68 (0.01)\\
		&DESeq2 & 0.13 (0.01)&0.88 (0.01)&0.60 (0.01)&0.84 (0.01)&0.80 (0.01)&0.76 (0.01)\\
		&Wilcoxon &0.15 (0.03)&0.80 (0.01)&0.91 (0.00)&0.88 (0.01)&0.92 (0.00)&0.83 (0.01)\\
		&Wilcoxon.TSS & 0.60 (0.03)&0.92 (0.01)&0.54 (0.03)&0.89 (0.01)&0.60 (0.02)&0.88 (0.01)\\
		&DACOMP & 0.09 (0.01)&0.71 (0.01)&0.37 (0.02)&0.61 (0.01)&0.74 (0.01)&0.59 (0.01)\\
		\hline\hline
	\end{tabular}
	\caption{Comparisons of different differential abundance tests when sequencing depth is unbalanced and the target is to control FDR. Standard error is reported in parenthesis. S1=Setting 1 and S2=Setting 2.}
	\label{tb:unbsdepfdr}
\end{table}

\begin{table}[h!]
	\small
	\centering
	\begin{tabular}{cccccccc}
		\hline\hline
		&& \multicolumn{2}{c}{$\rho=0.4$} & \multicolumn{2}{c}{$\rho=0.6$}  & \multicolumn{2}{c}{$\rho=0.8$}  \\ 
		&& FDR & Power &  FDR & Power & FDR & Power \\ 
		\hline
		\multirow{7}{*}{S1}&RDB & 0.06 (0.01)&0.77 (0.01)&0.06 (0.01)&0.74 (0.01)&0.06 (0.01)&0.75 (0.01)\\
		&ANCOM.BC &0.44 (0.03)&0.81 (0.01)&0.53 (0.03)&0.79 (0.01)&0.50 (0.03)&0.79 (0.01)\\
		&ANCOM & 0.27 (0.02)&0.77 (0.01)&0.35 (0.03)&0.76 (0.01)&0.32 (0.02)&0.75 (0.01)\\
		&DESeq2 & 0.61 (0.01)&0.83 (0.01)&0.63 (0.02)&0.81 (0.01)&0.63 (0.01)&0.81 (0.01)\\
		&Wilcoxon & 0.56 (0.04)&0.76 (0.01)&0.64 (0.04)&0.74 (0.01)&0.63 (0.04)&0.74 (0.01)\\
		&Wilcoxon.TSS & 0.57 (0.04)&0.77 (0.01)&0.67 (0.03)&0.76 (0.01)&0.66 (0.03)&0.75 (0.01)\\
		&DACOMP & 0.34 (0.03)&0.78 (0.01)&0.44 (0.03)&0.77 (0.01)&0.40 (0.03)&0.77 (0.01)\\
		\hline
		\multirow{7}{*}{S2}&RDB & 0.04 (0.01)&0.38 (0.01)&0.06 (0.01)&0.41 (0.01)&0.06 (0.01)&0.41 (0.01)\\
		&ANCOM.BC & 0.35 (0.03)&0.52 (0.01)&0.34 (0.03)&0.53 (0.01)&0.31 (0.03)&0.53 (0.01)\\
		&ANCOM & 0.22 (0.02)&0.46 (0.01)&0.22 (0.02)&0.46 (0.01)&0.19 (0.02)&0.46 (0.01)\\
		&DESeq2 & 0.54 (0.02)&0.55 (0.01)&0.51 (0.02)&0.56 (0.01)&0.49 (0.02)&0.56 (0.01)\\
		&Wilcoxon & 0.40 (0.04)&0.58 (0.02)&0.39 (0.04)&0.59 (0.02)&0.37 (0.04)&0.57 (0.02)\\
		&Wilcoxon.TSS & 0.43 (0.04)&0.61 (0.02)&0.41 (0.04)&0.61 (0.02)&0.36 (0.04)&0.59 (0.02)\\
		&DACOMP & 0.29 (0.03)&0.46 (0.01)&0.26 (0.03)&0.47 (0.01)&0.23 (0.03)&0.46 (0.01)\\
		\hline\hline
	\end{tabular}
	\caption{Comparisons of different differential abundance tests when there is correlation and the target is to control FDR. Standard error is reported in parenthesis. S1=Setting 1 and S2=Setting 2.}
	\label{tb:correlatedcasefdr}
\end{table}

\begin{table}[h!]
	\centering
	\begin{tabular}{cccccc}
		\hline\hline
		&& \multicolumn{2}{c}{$d=30$} & \multicolumn{2}{c}{$d=60$}  \\ 
		&& FDR & Power &  FDR & Power\\ 
		\hline
		\multirow{7}{*}{Setting 1}&RDB &0.20 (0.03)&0.96 (0.01)&0.18 (0.02)&0.97 (0.01)\\
		&ANCOM.BC & 0.13 (0.02)&0.98 (0.01)&0.11 (0.01)&0.97 (0.01)\\
		&ANCOM & 0.06 (0.01)&0.96 (0.01)&0.04 (0.01)&0.97 (0.01)\\
		&DESeq2 & 0.55 (0.02)&0.97 (0.01)&0.64 (0.01)&0.97 (0.01)\\
		&Wilcoxon & 0.20 (0.04)&0.88 (0.02)&0.24 (0.04)&0.89 (0.01)\\
		&Wilcoxon.TSS & 0.60 (0.03)&0.97 (0.01)&0.64 (0.02)&0.96 (0.01)\\
		&DACOMP & 0.08 (0.02)&0.95 (0.02)&0.11 (0.01)&0.94 (0.01)\\
		\hline
		\multirow{7}{*}{Setting 2}&RDB &0.14 (0.02)&0.91 (0.02)&0.16 (0.02)&0.75 (0.01)\\
		&ANCOM.BC & 0.12 (0.02)&0.94 (0.02)&0.10 (0.01)&0.75 (0.01)\\
		&ANCOM &0.05 (0.01)&0.91 (0.02)&0.05 (0.01)&0.72 (0.01)\\
		&DESeq2 &0.11 (0.02)&0.95 (0.02)&0.12 (0.01)&0.75 (0.01)\\
		&Wilcoxon & 0.13 (0.03)&0.78 (0.03)&0.11 (0.03)&0.65 (0.02)\\
		&Wilcoxon.TSS &0.39 (0.04)&0.94 (0.02)&0.45 (0.03)&0.79 (0.01)\\
		&DACOMP & 0.08 (0.02)&0.88 (0.02)&0.11 (0.02)&0.69 (0.01)\\
		\hline\hline
	\end{tabular}
	\caption{Comparisons of different differential abundance tests when number of components $d$ is small and the target is to control FDR. Standard error is reported in parenthesis.}
	\label{tb:smallnumtaxta}
\end{table}

{\it Covariate Balancing.} In this simulation study, we compare different methods' performance when some confounding variables are observed. Different from several previous simulation experiments, we adopt log-normal distributions with covariates here. We consider both setting 1 and 2, and set $s=20$, $m_1=m_2=50$, $d=200$ and $\beta=1$. We compare four methods: RDB test without any covariate adjustment, RDB test equipped with empirical balancing calibration weighting (CAL) by \cite{chan2016globally} (R package \texttt{ATE}), ANCOM, and ANCOM.BC. We focus on both FWER and FDR control in these simulation experiments. The results based on $100$ times simulation experiments are summarized in Table~\ref{tb:cvb}. From Table~\ref{tb:cvb}, the RDB test controls the FWER and FDR well when it works with weighting methods. 

\begin{table}[h!]
	\centering
	\begin{tabular}{ccccccc}
		\hline\hline
		&&& \multicolumn{2}{c}{Setting 1}  & \multicolumn{2}{c}{Setting 2} \\ 
		&&& FWER & Power  & FWER & Power  \\ 
		\hline
		\multirow{4}{*}{FWER Control}&RDB && 0.57 (0.05)&0.10 (0.02)&0.35 (0.05)&0.09 (0.01)\\
		&RDB-CAL && 0.06 (0.02)&0.63 (0.01)&0.11 (0.03)&0.44 (0.01)\\
		&ANCOM.BC && 0.50 (0.05)&0.79 (0.01)&0.52 (0.05)&0.66 (0.01)\\
		&ANCOM && 0.96 (0.02)&0.57 (0.01)&1.00 (0.00)&0.47 (0.01)\\
		\hline
		&&& \multicolumn{2}{c}{Setting 1}  & \multicolumn{2}{c}{Setting 2} \\ 
		&&& FDR & Power  & FDR & Power\\
		\hline
		\multirow{4}{*}{FDR Control}&RDB && 0.58 (0.03)&0.57 (0.03)&0.50 (0.03)&0.34 (0.03)\\
		&RDB-CAL && 0.16 (0.02)&0.79 (0.01)&0.14 (0.01)&0.60 (0.02)\\
		&ANCOM.BC && 0.31 (0.02)&0.85 (0.01)&0.27 (0.02)&0.76 (0.01)\\
		&ANCOM && 0.73 (0.01)&0.77 (0.01)&0.74 (0.01)&0.70 (0.01)\\
		\hline\hline
	\end{tabular}
	\caption{Comparisons of different differential abundance tests when covariates are observed. Standard error is reported in parenthesis.}
	\label{tb:cvb}
\end{table}

{\it Continous Outcome.} In this simulation experiment, we extend the binary outcome case to the continuous outcome. We still adopt Poisson-Gamma distribution but focus on the continuous outcome case. The setting of this simulation experiment is $m=40,80,160$, $s=20$, $d=200$ and $\beta=1$. Here, we compare the RDB test with ANCOM.BC, DESeq2, and DACOMP and aim to control FWER at the 10\% level. Table~\ref{tb:contoutcome} summarizes the simulation results of 100 simulation experiments. Through comparisons, we can conclude RDB test can control FWER very well. 

\begin{table}[h!]
	\small
	\centering
	\begin{tabular}{cccccccc}
		\hline\hline
		&& \multicolumn{2}{c}{$m=40$} & \multicolumn{2}{c}{$m=80$}  & \multicolumn{2}{c}{$m=160$}  \\ 
		&& FWER & Power &  FWER & Power & FWER & Power \\ 
		\hline
		\multirow{4}{*}{S1}&RDB &0.11 (0.03)&0.72 (0.01)&0.03 (0.02)&0.80 (0.01)&0.07 (0.03)&0.87 (0.01)\\
		&ANCOM.BC & 0.35 (0.05)&0.77 (0.01)&0.24 (0.04)&0.84 (0.01)&0.22 (0.04)&0.89 (0.01)\\
		&DESeq2 & 0.53 (0.05)&0.76 (0.01)&0.89 (0.03)&0.82 (0.01)&1.00 (0.00)&0.88 (0.01)\\
		&DACOMP & 0.13 (0.03)&0.43 (0.01)&0.09 (0.03)&0.61 (0.01)&0.19 (0.04)&0.74 (0.01)\\
		\hline
		\multirow{4}{*}{S2}&RDB & 0.12 (0.03)&0.49 (0.01)&0.08 (0.03)&0.61 (0.01)&0.05 (0.02)&0.71 (0.01)\\
		&ANCOM.BC & 0.41 (0.05)&0.54 (0.01)&0.29 (0.05)&0.66 (0.01)&0.16 (0.04)&0.74 (0.01)\\
		&DESeq2 & 0.05 (0.02)&0.54 (0.01)&0.15 (0.04)&0.69 (0.01)&0.17 (0.04)&0.77 (0.01)\\
		&DACOMP & 0.12 (0.03)&0.28 (0.01)&0.09 (0.03)&0.40 (0.01)&0.13 (0.03)&0.49 (0.01)\\
		\hline\hline
	\end{tabular}
	\caption{Comparisons of different differential abundance tests when the outcome of interest is a continuous variable. Standard error is reported in parenthesis. S1=Setting 1 and S2=Setting 2.}
	\label{tb:contoutcome}
\end{table}

{\it Distribution from Real Data.} In this simulation experiment, we investigate the performance of different tests when the data is drawn from some real dataset. In particular, we adopt the real data shuffling introduced in the previous section and set $m_1=m_2=50$. We compare RDB, ANCOM.BC,  DESeq2, Wilcoxon, Wilcoxon.TSS, and DACOMP in this simulation experiment. Both FWER and FDR control are studied and compared. We summarize the results based on 100 repeats in Table~\ref{tb:realdatashuffle}. The results suggest that most differential abundance tests can control false discoveries very well when we shuffle the real data. The idea of shuffling a real dataset can also be used to validate the choice of tuning parameters in the RDB test. 

\begin{table}[h!]
	\centering
	\begin{tabular}{ccccccc}
		\hline\hline
		& RDB & ANCOM.BC &  DESeq2 & Wilcoxon & Wilcoxon.TSS & DACOMP \\ 
		\hline
		FWER & 0.00 (0.00)&0.04 (0.02)&1.00 (0.00)&0.02 (0.01)&0.01 (0.01)&0.02 (0.01)\\
		FDR & 0.00 (0.00)&0.05 (0.02)&1.00 (0.00)&0.03 (0.02)&0.03 (0.02)&0.04 (0.02)\\
		\hline\hline
	\end{tabular}
	\caption{Comparisons of different differential abundance tests when the count data is randomly drawn from gut microbiome data. Standard error is reported in parenthesis. }
	\label{tb:realdatashuffle}
\end{table}

%%%%%%%%%%%%%%%%%%%%%%%%%%
\subsection{Choice of Parameters in Simulation Experiments}
%%%%%%%%%%%%%%%%%%%%%%%%%%

In this section, we report the choices of parameters for different methods.
\begin{enumerate}[label=(\alph*)]
	\item RDB: we always choose $M=\sqrt{2\log(d)/d}$, $q_\alpha=\sqrt{2\log d-2\log \alpha}$ with $\alpha=0.1$ and $r_Q=0.2$ for all experiments. When we want to control FDR, the critical value $T$ is chosen as Section~\ref{sc:fdr}.
	\item ANCOM.BC: we keep everything as default but set the level of significance as $\alpha=0.1$. We choose the Bonferroni method  to control FWER and BH method to control FDR. To handle the zero counts problem, we choose pseudo-count as 1.
	\item ANCOM: we keep everything as default but set the level of significance as $\alpha=0.1$. The cutoff proportion is chosen as 0.7 (use ``detected\_0.7") as recommended by the manual. We choose the Bonferroni method to control FWER and BH method to control FDR. To handle the zero counts problem, we choose pseudo-count as 1.
	\item DESeq2: we keep everything as default but set the significance level as $\alpha=0.1$ and fitType as ``mean". We choose the Bonferroni method to control FWER and BH method to control FDR. To handle the zero counts problem, we choose pseudo-count as 1.
	\item Wilcoxon: we choose the Bonferroni method to control FWER and BH method to control FDR.
	\item Wilcoxon.TSS: we choose the Bonferroni method to control FWER and BH method to control FDR.
	\item DACOMP: we select the reference set with function ``dacomp.select\_references" in the DACOMP package. All the settings in ``dacomp.select\_references" are default. We use ``dacomp.test" to test for differential abundance. In ``dacomp.test", the test option is chosen as ``DACOMP.TEST.NAME.WILCOXON" since it is a preferred option for 2-sample testing according to the manual. We also choose ``DACOMP.TEST.NAME.SPEARMAN" for the continuous outcome case. We set the significance level as $\alpha=0.1$ and keep other arguments in ``dacomp.test" as default. We choose the Bonferroni method to control FWER and BH method to control FDR. To handle the zero counts problem in reference selection, we choose pseudo-count as 1.
\end{enumerate}

%%%%%%%%%%%%%%%%%%%%%%%%%%
\subsection{Identified Genera by RDB test in Real Data Example}
%%%%%%%%%%%%%%%%%%%%%%%%%%

The genera identified by RDB test are reported in this section. 
\begin{enumerate}[label=(\alph*)]
	\item {\it MA vs. US.} Acidovorax, Akkermansia, Alistipes, Anaerotruncus, Bacteroides, Bilophila, Bulleidia, Butyrivibrio, Clostridium, Collinsella, Coprobacillus, Dehalobacterium, Dorea, Eggerthella, Eubacterium, Faecalibacterium, Herbaspirillum, Holdemania, Leptothrix, Lysobacter, Methanosphaera, Mobiluncus, Odoribacter, Oribacterium, Parabacteroides, Paucibacter, Pedobacter, Phascolarctobacterium, Prevotella, Roseburia, Ruminococcus, Sphaerotilus, Sphingobacterium, Sporanaerobacter, Succinivibrio, Tepidibacter and Turicibacter.
	\item {\it MA vs. VE.} None
	\item {\it US vs. VE.} Acetivibrio, Actinomyces, Akkermansia, Alistipes, Anaerococcus, Anaerofustis, Anaerostipes, Anaerotruncus, Anaerovorax, Bacteroides, Bifidobacterium, Bilophila, Blautia, Clostridium, Collinsella, Coprobacillus, Coprococcus, Corynebacterium, Dehalobacterium, Desulfitobacterium, Dialister, Dorea, Eggerthella, Eubacterium, Faecalibacterium, Granulicatella, Herbaspirillum, Holdemania, Lachnobacterium, Lachnospira, Leptothrix, Mobiluncus, Nitrosomonas, Odoribacter, Oribacterium, Oscillospira, Parabacteroides, Paucibacter, Phascolarctobacterium, Roseburia, Ruminococcus, Shuttleworthia, Sphaerotilus, Sphingobacterium, Sutterella, Tepidibacter and Turicibacter.
\end{enumerate}

%%%%%%%%%%%%%%%%%%%%%%%%%%
\subsection{Comparison with Other Methods in Analysis of Gut Microbiome Data}
\label{sc:compreal}
%%%%%%%%%%%%%%%%%%%%%%%%%%

\begin{figure}[h!]
	\centering
	\begin{subfigure}[b]{0.45\textwidth}
		\centering
		\includegraphics[width=\textwidth]{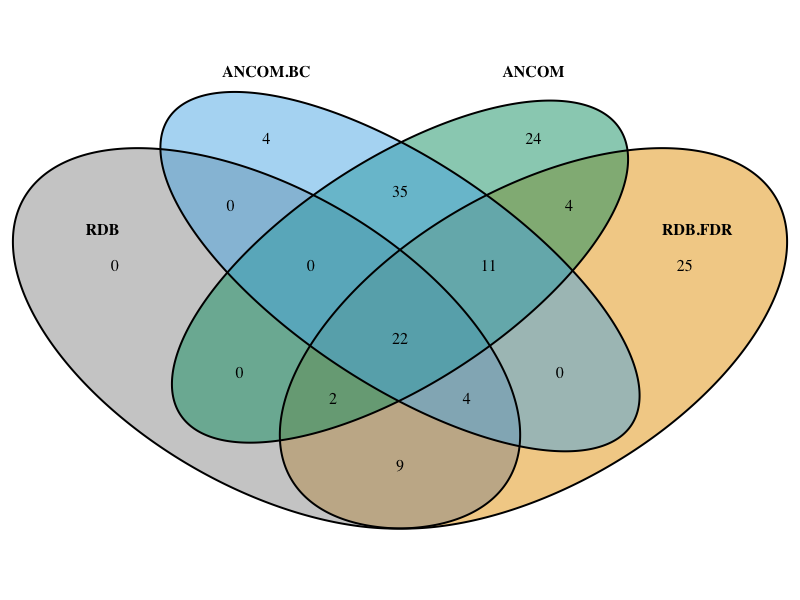}
		\caption{Malawi-USA}\label{fg:venn-MA-US}
	\end{subfigure}
	\hfill
	\begin{subfigure}[b]{0.45\textwidth}
		\centering
		\includegraphics[width=\textwidth]{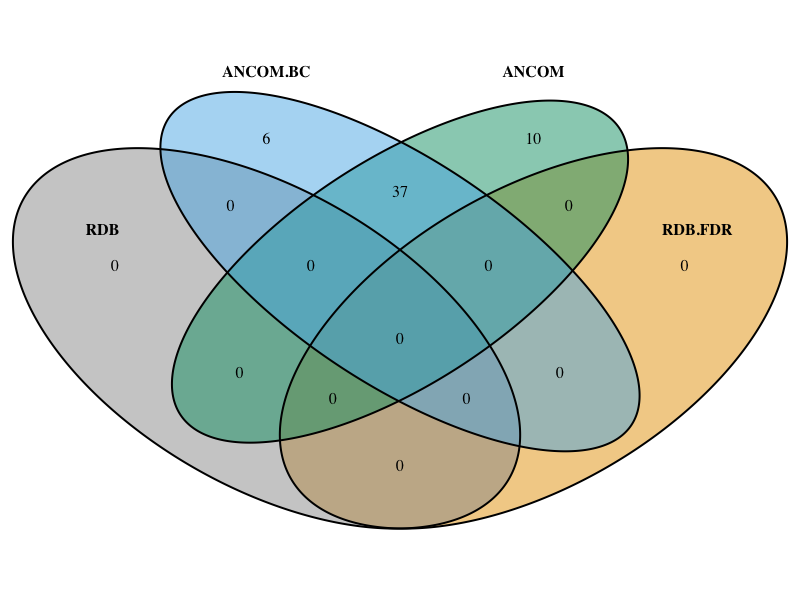}
		\caption{Malawi-Venezuela}\label{fg:venn-MA-VE}
	\end{subfigure}
	\begin{subfigure}[b]{0.45\textwidth}
		\centering
		\includegraphics[width=\textwidth]{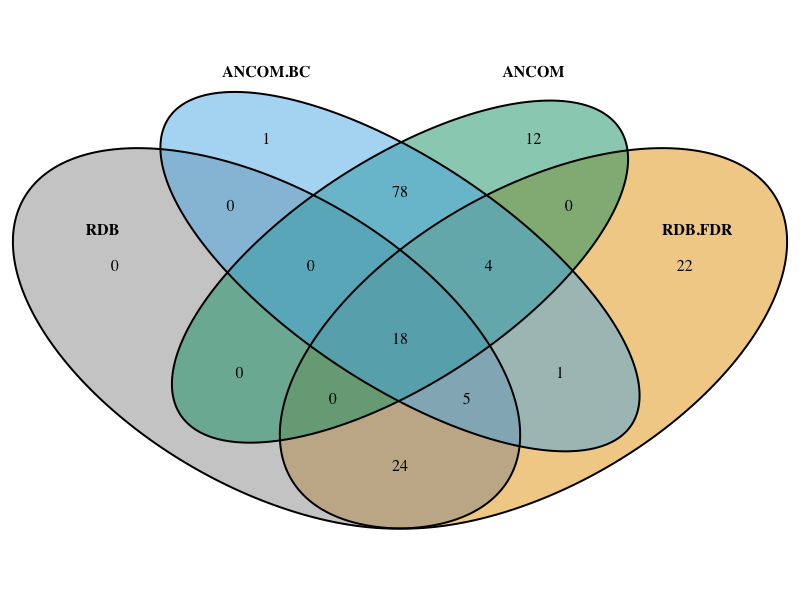}
		\caption{USA-Venezuela}\label{fg:venn-US-VE}
	\end{subfigure}
	\caption{Venn diagrams showing overlap between genera reported by different methods on the real data example.}\label{fg:overlap}
\end{figure}

In this section, we compare the RDB test with other methods in the real data example. In particular, we compare four different methods: RDB with FWER control (RDB), RDB with FDR control (RDB.FDR), ANCOM, and ANCOM.BC. The overlap between genera reported by these four different methods is shown in Figure~\ref{fg:overlap}. In Figure~\ref{fg:overlap}, RDB is relatively conservative as it aims for FWER control. The genera identified by RDB are usually reported by all other three methods or RDB.FDR alone as well. Besides the genera identified by RDB, there is a substantial difference between genera reported by RDB.FDR, ANCOM and ANCOM.BC. 

To further compare the results, we now investigate the differential and non-differential genera from the angle of reference-based hypothesis. The definition in the reference-based hypothesis suggests that we can regard all the non-differential genera identified by different methods as an estimated reference set. Ideally, given a perfect reference set, we can apply a standard two-sample test, such as a $t$-test, to identify differential genera after normalization on the reference set. So if the estimated reference set is good enough, the standard mean difference with respect to the reference set can separate the differential and non-differential genera very well. Figure~\ref{fg:realresult} shows the standard mean difference with respect to the reference set estimated by four different methods. In order to provide a consistent comparison, the genera are sorted according to the RDB's standard mean difference. Overall, these four methods' standard mean differences are highly correlated with each other. In RDB and RDB.FDR, the differential genera's standard mean differences are consistently different from the non-differential ones. However, some small standard mean differences in differential genera identified by ANCOM and ANCOM.BC suggest some of the non-differential genera might be misidentified. Therefore, the RDB test can better identify the differential genera from an angle of reference-based hypothesis. Compared with RDB, RDB.FDR can identify much more genera as it aims for FDR control. All these results are partly expected and are consistent with our simulation results.

\begin{figure}[h!]
	\centering
	\includegraphics[width=0.88\textwidth]{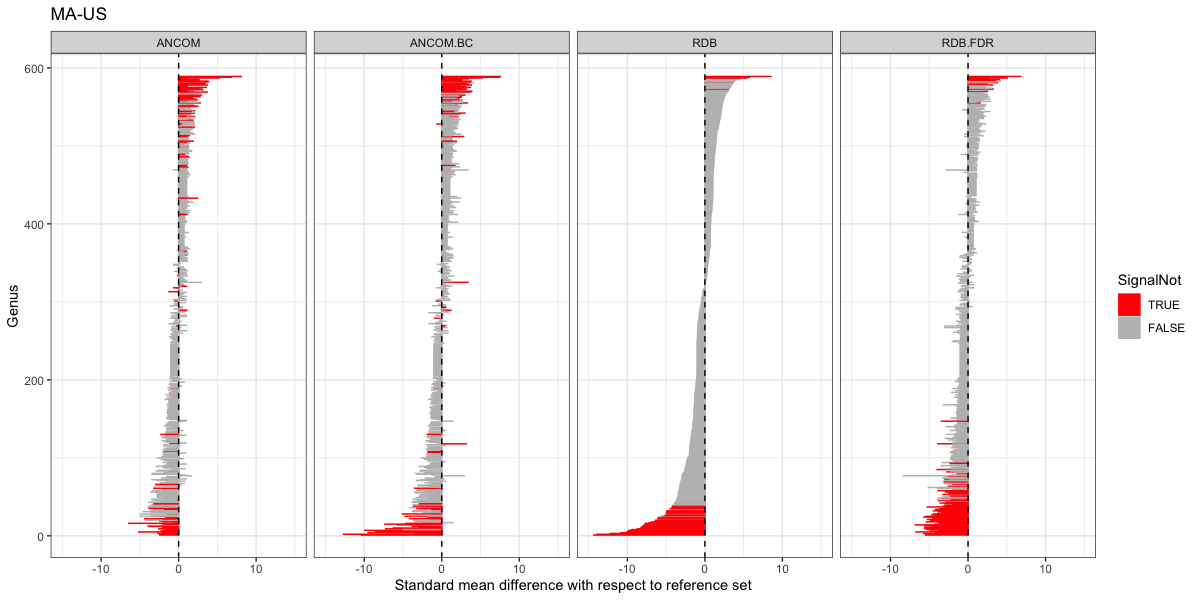}
	\includegraphics[width=0.88\textwidth]{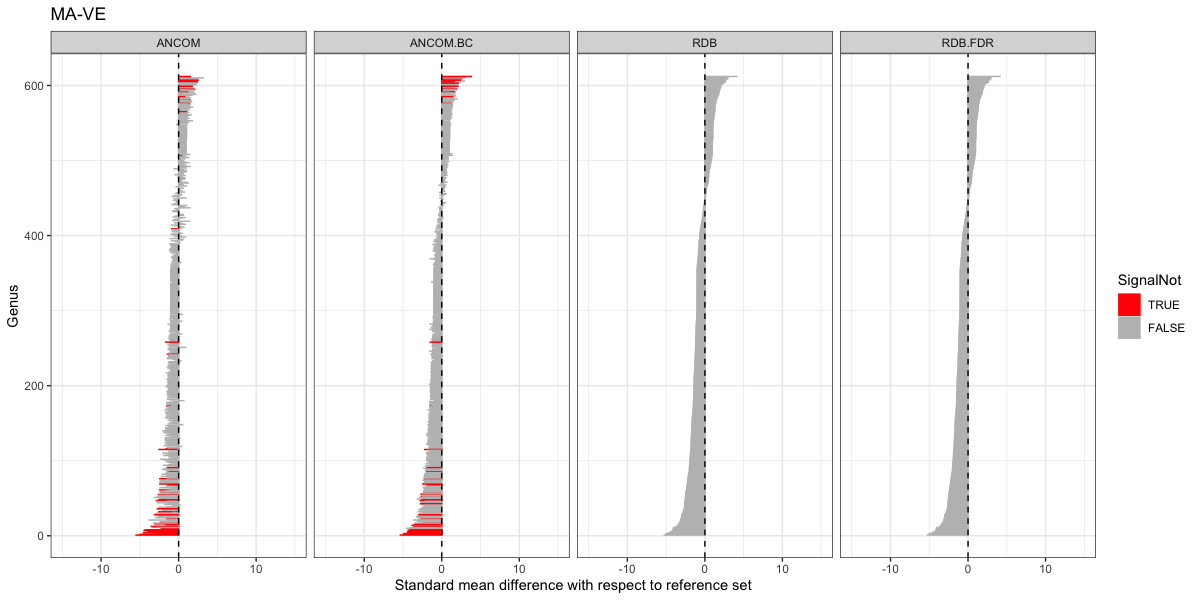}
	\includegraphics[width=0.88\textwidth]{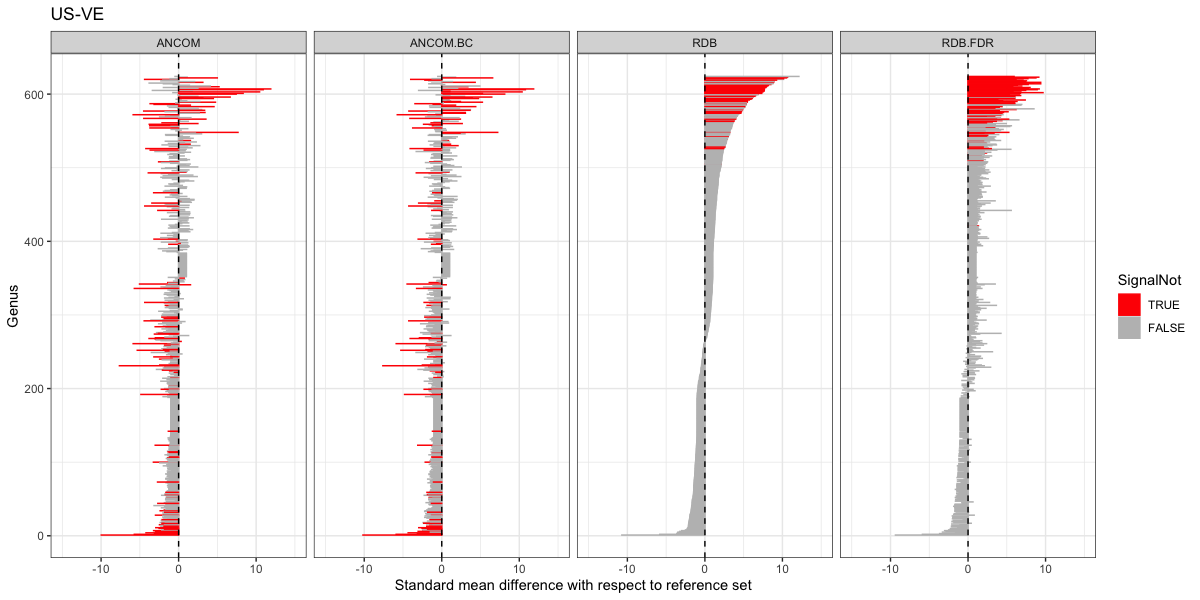}
	\caption{The standard mean difference with respect to the reference set estimated by different methods.}\label{fg:realresult}
\end{figure}

%%%%%%%%%%%%%%%%%%%%%%%%%%%%%%%%%%%%%%%%%%%%%%%
\section*{Appendix B}
%%%%%%%%%%%%%%%%%%%%%%%%%%%%%%%%%%%%%%%%%%%%%%%

In this appendix, we provide the proof for the main results, all the technical lemmas, and some technical assumptions.

%%%%%%%%%%%%%%%%%%%%%%%%%%
\subsection{Technical Assumptions}
%%%%%%%%%%%%%%%%%%%%%%%%%%

\begin{assumption}
	\label{assmp:center}
	If we define $\epsilon_{k,i}=\sqrt{m_k}{(\bar{P}_{k,i}-Q_{k,i})/\hat{\sigma}_{k,i}}$, then we assume the distribution of $\epsilon_{k,i}$, $ g_{k,i}$, satisfies 
	$$
	\int_{-\infty}^0 g_{k,i}(x)dx={1\over 2}.
	$$
\end{assumption}

\begin{assumption}
	\label{assmp:denupper}
	For any $\alpha$ and $i$, we assume the density of $\hat{R}_{i}(\alpha)$ is upper bounded by $U_f$.
\end{assumption}

\begin{assumption}
	\label{assmp:denlower}
	If $f_{i,\alpha}(x|s_1,s_2)$ is the conditional density function of  $r_i(\alpha)\epsilon_{1,i}-\sqrt{1-r^2_i(\alpha)}\epsilon_{2,i}$ given $\hat{\sigma}_{1,i}=s_1$ and $\hat{\sigma}_{2,i}=s_2$, then for any give $a$, we assume there is a constant $L_a>0$ such that
	$$
	f_{i,\alpha}(x|s_1,s_2)\ge L_a,\qquad -a<x<a.
	$$
\end{assumption}

\begin{assumption}
	\label{assmp:signaleven}
	We assume there exist a constant $C_r$ such that for any given $\alpha$,
	$$
	\max_{i\in I_0}{Q_{1,i}\over [\sigma_{1,i}^2/\alpha^2m_1+\sigma_{2,i}^2/(1-\alpha)^2m_2]^{1/2}}\le C_r{1\over |I_0|}\sum_{i\in I_0}{Q_{1,i}\over [\sigma_{1,i}^2/\alpha^2m_1+\sigma_{2,i}^2/(1-\alpha)^2m_2]^{1/2}}.
	$$
	In particular, when $\sigma_{1,i}=\sigma_{2,i}=\sigma_{i}$, we can just assume
	$$
	\max_{i\in I_0}{Q_{1,i}\over \sigma_{i}}\le C_r{1\over |I_0|}\sum_{i\in I_0}{Q_{1,i}\over \sigma_{i}}.
	$$
\end{assumption}

\begin{assumption}
	\label{assmp:samplesize}
	We assume $\log^3 d=o(m)$, where $m=\min(m_1,m_2)$.
\end{assumption}

\begin{assumption}
	\label{assmp:signal}
	We assume $|I_1|\le C_d\sqrt{d}$ for some constant $C_d$.
\end{assumption}

%%%%%%%%%%%%%%%%%%%%%%%%%%
\subsection{Proof of Theorem~\ref{thm:nonoise}}
%%%%%%%%%%%%%%%%%%%%%%%%%%

First, we show we always have $I_0\subset V_{(t)}$ and $U_{(t)}\subset I_1$ for each iteration $t$. Clearly, $I_{0}\subset V_{(0)}$ and $U_{(0)}\subset I_1$, since  $V_{(0)}=[d]$ and $U_{(0)}=\emptyset$. Now we assume $I_{0}\subset V_{(t)}$ and $U_{(t)}\subset I_1$. At iteration $t$, the definition of reference set suggests 
$$
{Q_{1,i} \over \sum_{i\in V_{(t)}}Q_{1,i}}=b(V_{(t)}){Q_{2,i} \over \sum_{i\in V_{(t)}}Q_{2,i}},\qquad i\in I_0.
$$
So $R_i(V_{(t)})$ for $i\in I_0$ have the same sign with $M(V_{(t)})$. This implies $W_{(t)}\cap I_0=\emptyset$, which immediately suggests that $I_{0}\subset V_{(t+1)}$ and $U_{(t+1)}\subset I_1$.

Next, we show that $U_{(T)}=I_1$ and $V_{(T)}=I_0$ when the loop stops. Suppose this is not the case, i.e. $V_{(T)}=I_0\cup I'_1$, where $I'_1\subset I_1$ and $I'_1\ne \emptyset$. Since $W_{(T-1)}$ is empty, we can know that $V_{(T-1)}=I_0\cup I'_1$. If $b(V_{(T-1)})>1$, there has to be at least one component $i\in I_1'$ such that 
$$
{Q_{1,i} \over \sum_{i\in V_{(T-1)}}Q_{1,i}}<{Q_{2,i} \over \sum_{i\in V_{(T-1)}}Q_{2,i}},
$$
because 
$$
\sum_{i\in V_{(T-1)}}{Q_{1,i} \over \sum_{i\in V_{(T-1)}}Q_{1,i}}=\sum_{i\in V_{(T-1)}} {Q_{2,i} \over \sum_{i\in V_{(T-1)}}Q_{2,i}}=1.
$$
So $W^-(V_{(T-1)})$ is not empty. Since $b(V_{(T-1)})>1$ implies $M(V_{(T-1)})>0$, we can know $W_{(T-1)}\ne \emptyset$. We can get the same conclusion when $b(V_{(T-1)})<1$. If $b(V_{(T-1)})=1$, then $W^+(V_{(T-1)})\cup W^-(V_{(T-1)})=I'_1$ and $W_{(T-1)}\ne \emptyset$. No matter what the value of $b(V_{(T-1)})$ is, we can know that $W_{(T-1)}\ne \emptyset$, which is contradicted with the stop condition $W_{(T-1)}= \emptyset$. Therefore, we show that $U_{(T)}=I_1$ and $V_{(T)}=I_0$. 

Lastly, as $W_{(t)}\ne \emptyset$ when $t=0,\ldots, T-2$, we can know that $T-1\le |I_1|$. The proof is complete.

%%%%%%%%%%%%%%%%%%%%%%%%%%
\subsection{Proof of Theorem~\ref{thm:fwer}}
%%%%%%%%%%%%%%%%%%%%%%%%%%
Throughout the proof, we write $m=\min(m_1,m_2)$. To account for the effect of renormalization, we define the following notation
$$
\hat{R}_i(\alpha)={\bar{P}_{1,i}/\alpha-\bar{P}_{2,i}/\sqrt{1-\alpha^2}\over [\hat{\sigma}_{1,i}^2/\alpha^2m_1+\hat{\sigma}_{2,i}^2/(1-\alpha)^2m_2]^{1/2}}
$$
For each subset $I$, there always exists an $\delta<\alpha_I<1-\delta$ such that $R_i(I)=R_i(\alpha_I)$. Here, $\hat{R}_i(\alpha)$ can also be rewritten as 
\begin{align*}
	\hat{R}_i(\alpha)&={Q_{1,i}/\alpha-Q_{2,i}/\sqrt{1-\alpha^2}\over [\hat{\sigma}_{1,i}^2/\alpha^2m_1+\hat{\sigma}_{2,i}^2/(1-\alpha)^2m_2]^{1/2}}+{\epsilon_{1,i}\hat{\sigma}_{1,i}/\sqrt{\alpha^2 m_1}-\epsilon_{2,i}\hat{\sigma}_{2,i}/\sqrt{(1-\alpha^2)m_2}\over [\hat{\sigma}_{1,i}^2/\alpha^2m_1+\hat{\sigma}_{2,i}^2/(1-\alpha)^2m_2]^{1/2}}\\
	&=r_i(\alpha)\left(\epsilon_{1,i}+{\sqrt{m_1}Q_{1,i}\over \hat{\sigma}_{1,i}}\right)-\sqrt{1-r^2_i(\alpha)}\left(\epsilon_{2,i}+{\sqrt{m_2}Q_{2,i}\over \hat{\sigma}_{2,i}}\right)\\
	&=r_i(\alpha)\epsilon'_{1,i}-\sqrt{1-r^2_i(\alpha)}\epsilon'_{2,i}.
\end{align*}
Here, $\epsilon_{k,i}=\sqrt{m_k}(\bar{P}_{k,i}-Q_{k,i})/\hat{\sigma}_{k,i}$, $\epsilon'_{k,i}=\epsilon_{k,i}+{\sqrt{m_k}Q_{k,i}/ \hat{\sigma}_{k,i}}$ and $r_i(\alpha)=(\hat{\sigma}_{1,i}/\sqrt{\alpha^2 m_1})/[\hat{\sigma}_{1,i}^2/\alpha^2m_1+\hat{\sigma}_{2,i}^2/(1-\alpha)^2m_2]^{1/2}$. We then can define the cumulative distribution function of a mixture distribution of  $\hat{R}_i(\alpha), i\in I_0$ for a given $\alpha$
$$
F^o_{\alpha}(x)={1\over |I_0|}\sum_{i\in I_0}F_{i,\alpha}(x),
$$
where $F_{i,\alpha}(x)$ is the cumulative distribution function of $\hat{R}_i(\alpha)$. The median of $F^o_{\alpha}(x)$ is denoted by $M^o_\alpha$, i.e. $M^o_\alpha={F^o_{\alpha}}^{-1}(1/2)$. Clearly, $M^o_\alpha$ is an indicator of $\alpha$. Specifically, $M^o_\alpha$ has the same sign with $b/\sqrt{1+b^2}-\alpha$, due to Assumption~\ref{assmp:center}. We write the order statistics of $\hat{R}_i(\alpha), i\in I_0$ as $\hat{R}_{(1)}(\alpha)\le \hat{R}_{(2)}(\alpha)\le \ldots\le \hat{R}_{(|I_0|)}(\alpha)$. Based on these order statistics, we define an event 
$$
\Acal_1=\left\{ L_R(\alpha,q)\le \hat{R}_{(\lfloor |I_0|q\rfloor)}(\alpha)\le U_R(\alpha,q), \delta<\alpha<1-\delta ,  |q-1/2|<\epsilon\right\},
$$
where $L_R(\alpha,q)={F^o_{\alpha}}^{-1}(q-A\sqrt{C_p\log d/ d})$, $U_R(\alpha,q)={F^o_{\alpha}}^{-1}(q+A\sqrt{C_p\log d/ d})$ and $A$ is specified later. An application of Lemma~\ref{lm:medianbound} suggests that $\PP(\Acal_1)=1-o(1)$ for a fixed $A$. Next, we can then define two events
$$
\Acal_2=\left\{ \sup_{i\in I_0,0<\alpha<1}\left|r_i(\alpha)\epsilon_{1,i}+\sqrt{1-r^2_i(\alpha)}\epsilon_{2,i}\right|\le q_\alpha\right\}
$$
and
$$
\Acal_3=\left\{ \sup_{i\in I_0,k=1,2}\left|{\hat{\sigma}_{k,i}-\sigma_{k,i}\over \sigma_{k,i}}\right|\le C_1\sqrt{\log d\over m} \right\}.
$$
Lemma~\ref{lm:supguassian} implies that $\PP(\Acal_2)\ge 1-\alpha+o(1)$ and an application of Lemma 3 in \cite{wang2020hypothesis} suggest that there exists a constant $C_1$ such that $\PP(\Acal_3)=1-o(1)$. The rest of analysis will be conducted conditioning on the event $\Acal=\Acal_1\cap \Acal_2\cap \Acal_3$.

When $|q-1/2|=o(\sqrt{\log d/d})$ and $|\alpha-b/\sqrt{1+b^2}|=o(1)$, then we can find a small $A$ such that
$$
|L_R(\alpha,q)-M^o_\alpha|\le \sqrt{C_p\log d\over d}\qquad {\rm and}\qquad |U_R(\alpha,q)-M^o_\alpha|\le \sqrt{C_p\log d\over d}.
$$
This suggested that conditioned on event $\Acal_1$, $|I_1|\le \sqrt{d}$ implies
$$
\begin{cases}
	M^o_{\alpha_I}<0&\qquad {\rm if}\ \hat{M}(I)<-M\\
	M^o_{\alpha_I}>0&\qquad {\rm if}\ \hat{M}(I)>M\\
	-2M<M^o_{\alpha_I}<2M&\qquad {\rm if}\ -M\le \hat{M}(I)\le M\\
\end{cases}
$$
for any $I$. At step $t$, there always exists a $\alpha_{(t)}$ such that $\hat{R}_i(V_{(t)})=\hat{R}_i(\alpha_{(t)})$. If $\hat{M}(V_{(t)})\ge M$, then we can know that $M_{\alpha_{(t)}}^o\ge 0$, which leads to $\alpha_{(t)}\le b/\sqrt{1+b^2}$. This suggests that for all $i\in I_0$,
$$
{Q_{1,i}/\alpha_{(t)}-Q_{2,i}/\sqrt{1-\alpha_{(t)}^2}\over [\hat{\sigma}_{1,i}^2/\alpha_{(t)}^2m_1+\hat{\sigma}_{2,i}^2/(1-\alpha_{(t)})^2m_2]^{1/2}}\ge 0.
$$
As $|r_i(\alpha_{(t)})\epsilon_{1,i}+\sqrt{1-r^2_i(\alpha_{(t)})}\epsilon_{2,i}|\le q_\alpha$ for all $i\in I_0$, we can conclude that 
$$
\hat{R}_i(V_{(t)})\ge -D_{(t)}^-,\qquad \forall\ i\in I_0.
$$
In other words, no false discovery is reported if $\hat{M}(V_{(t)})\ge M$. We can apply the similar argument to show that no false discovery is reported when $\hat{M}(V_{(t)})\le -M$.

If $-M< \hat{M}(V_{(t)})< M$, we have $-2M<M_{\alpha_{(t)}}^o< 2M$. By Lemma~\ref{lm:meanbound}, we can know 
$$
\max_{i\in I_0}\left|\EE\left({Q_{1,i}/\alpha_{(t)}-Q_{2,i}/\sqrt{1-\alpha_{(t)}^2}\over [\hat{\sigma}_{1,i}^2/\alpha_{(t)}^2m_1+\hat{\sigma}_{2,i}^2/(1-\alpha_{(t)})^2m_2]^{1/2}}I_{\Acal_3}\right)\right|\le {4C_rU_f\over L_{1/4}}\sqrt{C_p\log d\over d}.
$$
Since the analysis is conditioned on $\Acal_3$, we can know that
$$
\max_{i\in I_0}\left|{Q_{1,i}/\alpha_{(t)}-Q_{2,i}/\sqrt{1-\alpha_{(t)}^2}\over [\hat{\sigma}_{1,i}^2/\alpha_{(t)}^2m_1+\hat{\sigma}_{2,i}^2/(1-\alpha_{(t)})^2m_2]^{1/2}}\right|\le {8C_rU_f\over L_{1/4}}\sqrt{C_p\log d\over d}.
$$
Because $|r_i(\alpha_{(t)})\epsilon_{1,i}+\sqrt{1-r^2_i(\alpha_{(t)})}\epsilon_{2,i}|\le q_\alpha$, we have 
$$
-D_{(t)}^\pm\le \hat{R}_i(V_{(t)})\le D_{(t)}^\pm, \qquad \forall \ i\in I_0.
$$
Putting all together, no false discovery is reported at each step $t$ if the analysis is conditioned on event $\Acal$. We can complete the proof by noting that $\PP(\Acal)\ge 1-\alpha+o(1)$. 

%%%%%%%%%%%%%%%%%%%%%%%%%%
\subsection{Proof of Theorem~\ref{thm:power}}
%%%%%%%%%%%%%%%%%%%%%%%%%%

Without loss of generality, we assume $\delta^\ast>0$ in this proof. In the first iteration, we can know that
$$
|Q_{1,i}-Q_{2,i}|\ge (2+\epsilon)\sigma_{i}\sqrt{2\log d \over m},\qquad i\in I_1.
$$
If we define a event
$$
\Acal_1=\left\{\sup_{i\in [d],0<\alpha<1}\left|r_i(\alpha)\epsilon_{1,i}+\sqrt{1-r^2_i(\alpha)}\epsilon_{2,i}\right|\le q_\alpha\right\},
$$
then Lemma~\ref{lm:supguassian} suggests $\PP\left(\Acal_1\right)\to 1$. We can also define
$$
\Acal_2=\left\{ \sup_{i\in I_0,k=1,2}\left|{\hat{\sigma}_{k,i}-\sigma_{k,i}\over \sigma_{k,i}}\right|\le C_1\sqrt{\log d\over m} \right\}.
$$
Lemma 3 in \cite{wang2020hypothesis} suggest that there exists a constant $C_1$ such that $\PP(\Acal_2)=1-o(1)$. In the rest of analysis, we conduct the analysis conditioned on $\Acal_1\cap\Acal_2$.

When $-M<\hat{M}(V_{(0)})<M$, all components in $I_1$ are reported at once because for all $i\in I_1$, we have
$$
\left|\hat{R}_i(V_{(0)})\right|\ge (1+o(1))\sqrt{m}{|Q_{1,i}-Q_{2,i}|\over \sigma_i}-q_\alpha\ge (1+\epsilon)\sqrt{2\log d}\ge q_\alpha+r_QM.
$$
Thus, $I_1\subset \hat{I}_1$ if $-M<\hat{M}(V_{(0)})<M$. If $\hat{M}(V_{(0)})>M$, then all components in $I_1^-$ are reported at the first iteration as for all $i\in I_1^-$
$$
\hat{R}_i(V_{(0)})\le (1+o(1))\sqrt{m}{Q_{1,i}-Q_{2,i}\over \sigma_i}+q_\alpha\le -q_\alpha.
$$
With the similar analysis, we can know that no components in $I_1^+$ are reported. This suggests that $V_{(1)}=I'_0\cup I_1^+$, where $I'_0\subset I_0$. Because of $V_{(1)}$, we can know that the median of $M_{\alpha_{(1)}}^o\le 0$ and thus $-M<\hat{M}(V_{(1)})<M$ or $\hat{M}(V_{(1)})<-M$ happen with a probability approaching $1$. With the similar analysis, we can know that 
$$
\hat{R}_i(V_{(1)})\ge q_\alpha+r_QM,\qquad i\in I_1^+.
$$
So all components in $I_1^+$ are reported at the second iteration and we prove that $ I_1^-\cup I_1^+\subset \hat{I}_1$. 
If $\hat{M}(V_{(0)})<-M$, then we can apply the same arguments for the case of $\hat{M}(V_{(0)})>M$.

%%%%%%%%%%%%%%%%%%%%%%%%%%
\subsection{Proof of Proposition~\ref{prop:indetification}}
%%%%%%%%%%%%%%%%%%%%%%%%%%

Suppose there are two reference sets $I_{0,1}$ and $I_{0,2}$ such that $|I_{0,1}|>d/2$, $|I_{0,2}|>d/2$,
$$
Q_{1,i}=b_1Q_{2,i},\quad i\in I_{0,1}\quad {\rm and}\quad Q_{1,i}=b_2Q_{2,i},\quad i\in I_{0,2}
$$
for some positive number $b_1,b_2>0$. Since $|I_{0,1}|>d/2$ and $|I_{0,2}|>d/2$, we can know that $I_{0,1}\cap I_{0,2}\ne \emptyset$, which leads to $b_1=b_2$. Therefore, the null and alternative hypotheses based on $I_{0,1}$ and $I_{0,2}$ are the same. 

If we assume $|I_{0,1}|\le d/2$ and $|I_{0,2}|\le d/2$, then we can construct a case such that $I_{0,1}\cap I_{0,2}=\emptyset$,
$$
Q_{1,i}=b_1Q_{2,i},\quad i\in I_{0,1}\quad {\rm and}\quad Q_{1,i}=b_2Q_{2,i},\quad i\in I_{0,2}
$$
for some positive number $b_1\ne b_2$. If we use $I_{0,1}$ as reference set, then all components in $I_{0,1}$ belong to null hypothesis and all components in $I_{0,2}$ belong to alternative hypothesis. On the other hand, if we use $I_{0,2}$ as reference set, then all components in $I_{0,2}$ belong to null hypothesis and all components in $I_{0,1}$ belong to alternative hypothesis. Clearly, the null hypotheses defined by $I_{0,1}$ and $I_{0,2}$ are different. 

%%%%%%%%%%%%%%%%%%%%%%%%%%
\subsection{Technical Lemmas}
%%%%%%%%%%%%%%%%%%%%%%%%%%

\begin{lemma}
	\label{lm:meanbound}
	Suppose Assumption~\ref{assmp:center}-\ref{assmp:signaleven} are satisfied. There exists a small enough constant $\delta$ such that
	$$
	\max_{i\in I_0}\left|\EE\left({Q_{1,i}/\alpha-Q_{2,i}/\sqrt{1-\alpha^2}\over [\hat{\sigma}_{1,i}^2/\alpha^2m_1+\hat{\sigma}_{2,i}^2/(1-\alpha)^2m_2]^{1/2}}I_{\Acal_3}\right)\right|\le {2C_rU_f\over L_{1/4}}|M_\alpha^o|.
	$$
	for $|\alpha-b/\sqrt{1+b^2}|<\delta$. 
\end{lemma}
\begin{proof}
	For simplicity, we assume $\alpha<b/\sqrt{1+b^2}$ since the other case can be proved similarly. We write the distribution density function of $\hat{R}_i(\alpha)$ as $f_{i,\alpha}(x)$, which can be decomposed as
	$$
	f_{i,\alpha}(x)=\int f_{i,\alpha}(x-T_{i,\alpha}|s_1,s_2)\pi(s_1,s_2)ds_1ds_2,
	$$
	where $T_{i,\alpha}=r_i(\alpha)\sqrt{m_1}Q_{1,i}/\hat{\sigma}_{1,i}-\sqrt{1-r^2_i(\alpha)}\sqrt{m_2}Q_{2,i}/\hat{\sigma}_{2,i}$, $f_{i,\alpha}(x|s_1,s_2)$ is the conditional density function of  $r_i(\alpha)\epsilon_{1,i}-\sqrt{1-r^2_i(\alpha)}\epsilon_{2,i}$ given $\hat{\sigma}_{1,i}=s_1$ and $\hat{\sigma}_{2,i}=s_2$ and $\pi(s_1,s_2)$ is distribution density function of $\hat{\sigma}_{1,i}$ and $\hat{\sigma}_{2,i}$. The definition of $M_\alpha^o$ suggests that 
	$$
	{1\over |I_0|}\sum_{i\in I_0} \int_{-\infty}^{M_\alpha^o}\left(\int f_{i,\alpha}(x-T_{i,\alpha}|s_1,s_2)\pi(s_1,s_2)ds_1ds_2\right) dx={1\over 2}
	$$
	As $\int_{-\infty}^0g^o_{k,i}(x)dx=1/2$, we can know $\int_{-\infty}^0f_{i,\alpha}(x|s_1,s_2)dx=1/2$. This suggests that
	$$
	{1\over |I_0|}\sum_{i\in I_0} \int_{-\infty}^{0}\left(\int f_{i,\alpha}(x|s_1,s_2)\pi(s_1,s_2)ds_1ds_2\right) dx={1\over 2},
	$$
	which leads to
	\begin{align*}
		&{1\over |I_0|}\sum_{i\in I_0} \int_{0}^{M_\alpha^o}\left(\int f_{i,\alpha}(x-T_{i,\alpha}|s_1,s_2)\pi(s_1,s_2)ds_1ds_2\right) dx\\
		=& {1\over |I_0|}\sum_{i\in I_0} \int_{-\infty}^{0}\left(\int \left[f_{i,\alpha}(x|s_1,s_2)-f_{i,\alpha}(x-T_{i,\alpha}|s_1,s_2)\right]\pi(s_1,s_2)ds_1ds_2\right) dx.
	\end{align*}
	The left hand can be upper bounded by 
	$$
	{1\over |I_0|}\sum_{i\in I_0} \int_{0}^{M_\alpha^o}\left(\int f_{i,\alpha}(x-T_{i,\alpha}|s_1,s_2)\pi(s_1,s_2)ds_1ds_2\right) dx\le U_f M_\alpha^o
	$$
	because the density of $\hat{R}_i(\alpha)$ is upper bounded by $U_f$. The right hand can be lower bounded by 
	\begin{align*}
		& {1\over |I_0|}\sum_{i\in I_0} \int_{-\infty}^{0}\left(\int \left[f_{i,\alpha}(x|s_1,s_2)-f_{i,\alpha}(x-T_{i,\alpha}|s_1,s_2)\right]\pi(s_1,s_2)ds_1ds_2\right) dx\\
		=& {1\over |I_0|}\sum_{i\in I_0} \int \left[F_{i,\alpha}(0|s_1,s_2)-F_{i,\alpha}(-T_{i,\alpha}|s_1,s_2)\right]\pi(s_1,s_2)ds_1ds_2\\
		\ge & {1\over |I_0|}\sum_{i\in I_0} \int_{\Acal_3} T_{i,\alpha} f_{i,\alpha}(-T_{i,\alpha}|s_1,s_2)\pi(s_1,s_2)ds_1ds_2\\
		\ge &     L_{1/4}{1\over |I_0|}\sum_{i\in I_0}\EE(T_{i,\alpha}I_{\Acal_3})
	\end{align*}
	Here, we know $|T_{i,\alpha}|\le 1/4$ when $\delta$ is small enough and it is conditioned on $\Acal_3$. The Assumption~\ref{assmp:signaleven} suggests that 
	$$
	\max_{i\in I_0}\EE(T_{i,\alpha}I_{\Acal_3})\le 2C_r{1\over |I_0|}\sum_{i\in I_0}\EE(T_{i,\alpha}I_{\Acal_3})
	$$
	Now, we can then conclude
	$$
	\max_{i\in I_0}\EE(T_{i,\alpha}I_{\Acal_3})\le {2C_rU_f\over L_{1/4}}M_\alpha^o.
	$$	
\end{proof}

\begin{lemma}
	\label{lm:supguassian}
	If $\log^3 d=o(m)$, then
	$$
	\PP\left(\sup_{i\in I_0,0<\alpha<1}\left|r_i(\alpha)\epsilon_{1,i}+\sqrt{1-r^2_i(\alpha)}\epsilon_{2,i}\right|>q_\alpha\right)\le \alpha+o(1).
	$$
\end{lemma}
\begin{proof}
	We first define 
	$$
	S_i(\alpha)={(\bar{P}_{1,i}-Q_{1,i})/\alpha-(\bar{P}_{2,i}-Q_{2,i})/\sqrt{1-\alpha^2}\over [\sigma_{1,i}^2/\alpha^2m_1+\sigma_{2,i}^2/(1-\alpha)^2m_2]^{1/2}}
	$$
	Clearly, 
	\begin{align*}
		&\left|\sup_{i\in I_0,0<\alpha<1}\left|r_i(\alpha)\epsilon_{1,i}+\sqrt{1-r^2_i(\alpha)}\epsilon_{2,i}\right|-\sup_{i\in I_0,0<\alpha<1}\left|S_i(\alpha)\right|\right|\\
		\le & \sup_{i\in I_0,0<\alpha<1}\left|S_i(\alpha)\right|\sup_{i\in I_0,k=1,2}\left|{\hat{\sigma}_{k,i}-\sigma_{k,i}\over \sigma_{k,i}}\right|.
	\end{align*}
	An application of Lemma 3 in \cite{wang2020hypothesis} suggest that
	$$
	\PP\left(\sup_{i\in I_0,k=1,2}\left|{\hat{\sigma}_{k,i}-\sigma_{k,i}\over \sigma_{k,i}}\right|\ge C_1\sqrt{\log d\over m} \right)\to 0.
	$$
	By union bound, we have
	$$
	\PP\left(\sup_{i\in I_0,0<\alpha<1}|S_i(\alpha)|>q_\alpha-2C_1{\log d\over \sqrt{m}}\right)\le \sum_{i\in I_0}\PP\left(\sup_{0<\alpha<1}|S_i(\alpha)|>q_\alpha-2C_1{\log d\over \sqrt{m}}\right)
	$$
	Note that $P_{k,j,i}$ is naturally sub-Gaussian random variable, as it is bounded. By Theorem 1.1 in \cite{zaitsev1987gaussian}, we have 
	\begin{align*}
		&\PP\left(\sup_{0<\alpha<1}|S_i(\alpha)|>q_\alpha-2C_1{\log d\over \sqrt{m}}\right)\\
		\le &\PP\left(\sup_{0<r_i<1}|r_iz_{1,i}+\sqrt{1-r_i^2}z_{2,i}|>q_\alpha-(2C_1+2/C_3){\log d\over \sqrt{m}}\right)+C_2\exp\left(-{C_3(2/C_3)\log d}\right)
	\end{align*}
	Here, $z_{1,i}$ and $z_{2,i}$ are standard normal distribution. Observe that
	\begin{align*}
		&\PP\left(\sup_{0<r_i<1}|r_iz_{1,i}+\sqrt{1-r_i^2}z_{2,i}|>q_\alpha-(2C_1+2/C_3){\log d\over \sqrt{m}}\right)\\
		\le & \PP\left(z^2_{1,i}+z^2_{2,i}>\left(q_\alpha-(2C_1+2/C_3){\log d\over \sqrt{m}}\right)^2\right)\\
		\le & \exp\left(-\left(q_\alpha-(2C_1+2/C_3){\log d\over \sqrt{m}}\right)^2/2\right)\\
		=&(1+o(1))\alpha/d
	\end{align*}
	Putting everything together yields
	$$
	\PP\left(\sup_{i\in I_0,0<\alpha<1}\left|r_i(\alpha)\epsilon_{1,i}+\sqrt{1-r^2_i(\alpha)}\epsilon_{2,i}\right|>q_\alpha\right)\le \alpha+o(1)
	$$
\end{proof}

\begin{lemma}
	\label{lm:medianbound}
	If  \eqref{eq:dependency} is satisfied, then there exists constant $c_1$ and $c_2$ such that 
	$$
	\PP\left(\sup_{\delta<\alpha<1-\delta ,  |q-1/2|<\epsilon}\left[\hat{R}_{(\lfloor |I_0|q\rfloor)}(\alpha)-{F^o_{\alpha}}^{-1}\left(q+\sqrt{AC_p\log d/ d}\right)\right]>0\right)\le {c_1\over d^{A/c_2}}
	$$
	and 
	$$
	\PP\left(\inf_{\delta<\alpha<1-\delta ,  |q-1/2|<\epsilon}\left[\hat{R}_{(\lfloor |I_0|q\rfloor)}(\alpha)-{F^o_{\alpha}}^{-1}\left(q-\sqrt{AC_p\log d/ d}\right)\right]<0\right)\le {c_1\over d^{A/c_2}}.
	$$
\end{lemma}
\begin{proof}
	It is sufficient to show the first conclusion here, as the same arguments can also be applied to the second conclusion. We write
	\begin{align*}
		&\PP\left(\sup_{\delta<\alpha<1-\delta ,  |q-1/2|<\epsilon}\left[\hat{R}_{(\lfloor |I_0|q\rfloor)}(\alpha)-{F^o_{\alpha}}^{-1}\left(q+\sqrt{AC_p\log d/ d}\right)\right]>0\right)\\
		=&\PP\left(\bigcup_{\delta<\alpha<1-\delta  ,  |q-1/2|<\epsilon}\left\{\hat{R}_{(\lfloor |I_0|q\rfloor)}(\alpha)-{F^o_{\alpha}}^{-1}\left(q+\sqrt{AC_p\log d/ d}\right)>0\right\}\right)\\
		=& \PP\left(\bigcup_{\delta<\alpha<1-\delta  ,  |q-1/2|<\epsilon}\left\{\sum_{i\in I_0}B_{i}(\alpha,q)\le \lfloor |I_0|q\rfloor\right\}\right)\\
		=& \PP\left(\inf_{\delta<\alpha<1-\delta  ,  |q-1/2|<\epsilon}\left[\sum_{i\in I_0}B_{i}(\alpha,q)-\mu_{i}(\alpha,q)\right]\le -\sqrt{AC_pd\log d}\right)\\
		\le & \PP\left(\sup_{\delta<\alpha<1-\delta  ,  |q-1/2|<\epsilon}\left|{1\over |I_0|}\sum_{i\in I_0}B'_{i}(\alpha,q)\right|\ge {\sqrt{AC_p\log d} \over\sqrt{d}}\right),
	\end{align*}
	where $B_{i}(\alpha,q)=I(\hat{R}_i(\alpha)<{F^o_{\alpha}}^{-1}(q+{AC_p\log d/ \sqrt{d}}))$, $\mu_{i}(\alpha,q)=\EE\left(B_{i}(\alpha,q)\right)$ and $B'_{i}(\alpha,q)=B_{i}(\alpha,q)-\mu_{i}(\alpha,q)$. It is clear that $|I_0|^{-1}\sum_{i\in I_0}\mu_{i}(\alpha,q)=q+AC_p\log d/ \sqrt{d}$. To investigate the supremum of the above stochastic process, we define the set of functions of interest
	\begin{equation}
		\label{eq:set}
		\Bcal:=\left\{\sum_{i\in I_0}B_{i}(\alpha,q): \delta<\alpha<1-\delta ,  |q-1/2|<\epsilon \right\}.
	\end{equation}
	For any $(\alpha,q)$ and $(\alpha',q')$, we can define a distance between them as 
	\begin{equation}
		\label{eq:distance}
		D_\Delta\left(\sum_{i\in I_0}B_{i}(\alpha,q),\sum_{i\in I_0}B_{i}(\alpha',q')\right)={1\over |I_0|}\sum_{i\in I_0}\EE\left(|B_{i}(\alpha,q)-B_{i}(\alpha',q')|\right).
	\end{equation}
	The Lemma~\ref{lm:entropy} suggests that there exist $\{B_j^L,B_j^U\}_{j=1}^N$ with $N\le 36\epsilon d^2$ such that $D_\Delta(B_j^L,B_j^U)\le 1/d$ and such that for all $B\in \Bcal$, there is a $j$ such that $B_j^L\le B\le B_j^U$. Thus, there exists $\{\alpha_j,q_j\}_{j=1}^N$ such that 
	$$
	\sup_{\delta<\alpha<1-\delta  ,  |q-1/2|<\epsilon}\left|{1\over |I_0|}\sum_{i\in I_0}B'_{i}(\alpha,q)\right|\le \max_{j=1,\ldots, N}\left|{1\over |I_0|}\sum_{i\in I_0}B'_{i}(\alpha_j,q_j)\right|+\max_{j=1,\ldots, N}\left|{1\over |I_0|}(B_j^U-B_j^L)\right|.
	$$
	An application of union bound and Lemma~\ref{lm:concen} suggest that the second above term can be bounded well
	$$
	\PP\left(\max_{j=1,\ldots, N}\left|{1\over |I_0|}(B_j^U-B_j^L)\right|>C_1{\log d+\sqrt{2C_p\log d} \over d}\right)\le {1\over d^{\min(C_1^2,C_1)}}.
	$$
	We then apply the chaining argument to bound the first term \cite{talagrand2014upper,wang2020hypothesis,wang2019structured}. Specifically, we apply Theorem 2.2.27 in \cite{talagrand2014upper} to obtain
	$$
	\PP\left(\max_{j=1,\ldots, N}\left|{1\over |I_0|}\sum_{i\in I_0}B'_{i}(\alpha_j,q_j)\right|>{C_2+C_3t\over \sqrt{d}}\right)\le C_4 \exp(-t^2/C_p).
	$$
	If we choose $t=\sqrt{0.8AC_p\log d/C_3}$, we can have
	$$
	\PP\left(\max_{j=1,\ldots, N}\left|{1\over |I_0|}\sum_{i\in I_0}B'_{i}(\alpha_j,q_j)\right|>{\sqrt{0.9AC_p\log d}\over \sqrt{d}}\right)\le {C_4\over d^{0.8A/C_3}} .
	$$
	When $d$ is large enough, putting two term together yields
	$$
	\PP\left(\sup_{\delta<\alpha<1-\delta  ,  |q-1/2|<\epsilon}\left|{1\over |I_0|}\sum_{i\in I_0}B'_{i}(\alpha,q)\right|>{\sqrt{AC_p\log d}\over \sqrt{d}}\right)\le {C_4+1\over d^{0.8A/C_3}} .
	$$
	We complete the proof. 
\end{proof}

\begin{lemma}
	\label{lm:concen}
	If \eqref{eq:dependency} is satisfied, then given $(\alpha,q)$, we have 
	$$
	\PP\left(\left|\sum_{i\in I_0}B'_{i}(\alpha,q)\right|>t\right)\le \exp\left(-{t^2\over 2C_p\sum_{i\in I_0}\mu_{i}(\alpha,q) +2t/3}\right).
	$$
	For any $(\alpha,q)$ and $(\alpha',q')$, we have 
	$$
	\PP\left(\left|\sum_{i\in I_0}B'_{i}(\alpha,q)-B'_{i}(\alpha',q')\right|>t\right)\le 2\exp\left(-{t^2\over 2C_p|I_0|D_\Delta +2t/3}\right),
	$$
	where $D_\Delta=D_\Delta\left(\sum_{i\in I_0}B_{i}(\alpha,q),\sum_{i\in I_0}B_{i}(\alpha',q')\right)$.
\end{lemma}
\begin{proof}
	Given any $S_i$, the condition \eqref{eq:dependency} and Theorem 4 in \cite{delyon2009exponential} suggests that
	\begin{equation}
		\label{eq:conin}
		\PP\left(\left|\sum_{i\in I_0}B_{i}-\mu_i\right|>t\right)\le \exp\left(-{t^2\over 2C_p\sum_{i\in I_0}\mu_{i} +2t/3}\right).
	\end{equation}
	This immediately leads to the first concentration inequality. To show the second concentration inequality, we define the function
	$$
	B_i^+=\begin{cases}
		1,&\qquad  B_{i}(\alpha,q)=1\ {\rm and}\ B_{i}(\alpha',q')=0\\
		0,&\qquad  {\rm otherwise}
	\end{cases}
	$$
	and 
	$$
	B_i^-=\begin{cases}
		1,&\qquad  B_{i}(\alpha,q)=0\ {\rm and}\ B_{i}(\alpha',q')=1\\
		0,&\qquad  {\rm otherwise}.
	\end{cases}
	$$
	Clearly, $\sum_{i\in I_0}B'_{i}(\alpha,q)-B'_{i}(\alpha',q')=\sum_{i\in I_0}{B_i^+}'-{B_i^-}'$. Applying \eqref{eq:conin} twice yields the second concentration inequality.	
\end{proof}

\begin{lemma}
	\label{lm:entropy}
	With $\Bcal$ and distance $D_\Delta$ defined in \eqref{eq:set} and \eqref{eq:distance}, the covering number with bracketing satisfies
	$$
	N_B(\Bcal,D_\Delta,t)\le {36\epsilon\over t^2}.
	$$
	Here, $N_B(\Bcal,D_\Delta,t)$ is defined as the smallest value of $N$ for which there exist pairs of functions $\{B_j^L,B_j^U\}_{j=1}^N$ such that $D_\Delta(B_j^L,B_j^U)\le t$ and such that for all $B\in \Bcal$, there is a $j$ such that $B_j^L\le B\le B_j^U$.
\end{lemma}
\begin{proof}
	We define the event $E_i(\alpha,q)=\{\hat{R}_i(\alpha)<{F^o_{\alpha}}^{-1}(q+{AC_p\log d/ \sqrt{d}})\}$. Clearly, $\PP(E_i(\alpha,q))=\EE(B_i(\alpha,q))$. We consider $q_0,\ldots,q_{N_q}$ such that
	$$
	q_k=1/2-\epsilon+kt/3,\qquad {\rm and}\qquad N_q=\lfloor {6\epsilon/t}\rfloor.
	$$ 
	For each $q_k$, we can consider a sequence of $\alpha$, i.e. $\delta=\alpha_{k,0}<\alpha_{k,1}<\ldots<\alpha_{k,n_k}<1-\delta$. Given  $\alpha_{k,j}$, we define  a function of $\alpha'$
	$$
	L(\alpha'):={1\over |I_0|}\sum_{i\in I_0}\PP\left(\bigcup_{\alpha_{k,j}\le\alpha<\alpha'} E_{i}(\alpha,q_k)\setminus\bigcap_{\alpha_{k,j}\le\alpha<\alpha'} E_{i}(\alpha,q_k)\right).
	$$
	Clearly, $L(\alpha')$ is a non-decreasing function of $\alpha'$ when $\alpha'>\alpha_{k,j}$.  Then, $\alpha_{k,j+1}$ is chosen as the smallest number such that $L(\alpha')=t/3$. We also choose $\alpha_{k,n_k}$ such that $\alpha_{k,n_k+1}\ge 1-\delta$. Given $\alpha_{k,j}$ and $\alpha_{k,j+1}$, we define events
	$$
	E_{i,k,j}=\bigcup_{\alpha_{k,j}\le\alpha<\alpha_{k,j+1}} E_{i}(\alpha,q_k)\setminus\bigcap_{\alpha_{k,j}\le\alpha<\alpha_{k,j+1}} E_{i}(\alpha,q_k)
	$$
	$$
	E_{i,k,j}^+=E_{i,k,j}\bigcap\left\{\left(r_i(\alpha_{k,j})-r_i(\alpha_{k,j+1})\right)\epsilon'_{1,i}\ge \left(\sqrt{1-r^2_i(\alpha_{k,j+1})}-\sqrt{1-r^2_i(\alpha_{k,j})}\right)\epsilon'_{2,i}\right\}
	$$
	and 
	$$
	E_{i,k,j}^-=E_{i,k,j}\bigcap\left\{\left(r_i(\alpha_{k,j})-r_i(\alpha_{k,j+1})\right)\epsilon'_{1,i}< \left(\sqrt{1-r^2_i(\alpha_{k,j+1})}-\sqrt{1-r^2_i(\alpha_{k,j})}\right)\epsilon'_{2,i}\right\}.
	$$
	Clearly, $E_{i,k,j}^+$ are disjoint for $0\le j\le n_k$ and $E_{i,k,j}^-$ are disjoint for $0\le j\le n_k$, which leads to 
	$$
	\sum_{j=0}^{n_k}\PP(E_{i,k,j}^+)\le 1\qquad {\rm and}\qquad  \sum_{j=0}^{n_k}\PP(E_{i,k,j}^-)\le 1.
	$$
	This suggests that 
	$$
	\sum_{j=0}^{n_k}{1\over |I_0|}\sum_{i\in I_0}\PP(E_{i,k,j})\le 2
	$$
	Since the choice of $\alpha_{k,j}$ suggest that $|I_0|^{-1}\sum_{i\in I_0}\PP(E_{i,k,j})=t/3$, we can conclude that $n_k\le 6/t$.
	
	Then, we consider a partition of $\Bcal$, i.e. $\Bcal_{k,j}$ is
	$$
	\Bcal_{k,j}=\left\{\sum_{i\in I_0}B_{i}(\alpha,q): q_k\le q<q_{k+1},\alpha_{k,j}\le \alpha< \alpha_{k,j+1}\right\}.
	$$
	After we define the sets
	$$
	E^U_{i,k,j}=\bigcup_{q_k\le q<q_{k+1},\alpha_{k,j}\le \alpha< \alpha_{k,j+1}}E_i(\alpha, q)\qquad {\rm and}\qquad E^L_{i,k,j}=\bigcap_{q_k\le q<q_{k+1},\alpha_{k,j}\le \alpha< \alpha_{k,j+1}}E_i(\alpha, q),
	$$
	we then define functions
	$$
	B_{k,j}^U=\sum_{i\in I_0}I(E^U_{i,k,j})\qquad {\rm and}\qquad B_{k,j}^L=\sum_{i\in I_0}I(E^L_{i,k,j}).
	$$
	By the definition, we can know that for any $B\in \Bcal_{k,j}$, $B_{k,j}^L\le B\le B_{k,j}^U$. Furthermore, as
	$$
	E^U_{i,k,j}\setminus E^L_{i,k,j}\subset E_{i,k,j}\bigcup\left(E_i(\alpha_{k,j}, q_{k+1})\setminus E_i(\alpha_{k,j}, q_k)\right)\bigcup\left(E_i(\alpha_{k,j+1}, q_{k+1})\setminus E_i(\alpha_{k,j+1}, q_k)\right),
	$$
	we have 
	\begin{align*}
		&D_\Delta(B_{k,j}^L,B_{k,j}^U)\\
		=&{1\over |I_0|}\sum_{i\in I_0}\PP\left(E^L_{i,k,j}\setminus E^U_{i,k,j}\right)\\
		\le&{1\over |I_0|}\sum_{i\in I_0}\PP\left(E_{i,k,j}\right)+\PP(E_i(\alpha_{k,j}, q_{k+1})\setminus E_i(\alpha_{k,j}, q_k))+\PP(E_i(\alpha_{k,j+1}, q_{k+1})\setminus E_i(\alpha_{k,j+1}, q_k))\\
		\le & {t\over 3}+{t\over 3}+{t\over 3}\\
		\le & t
	\end{align*}
	Now, we can conclude that 
	$$
	N_B(\Bcal,D_\Delta,t)\le {6\epsilon\over t}\times {6\over t}={36\epsilon\over t^2}.
	$$
\end{proof}

\end{document}